\documentclass[journal,10pt]{IEEEtran}

\usepackage{setspace}
\usepackage{amsthm}

\usepackage{amsmath}%
\usepackage{MnSymbol}%
\usepackage{wasysym}
\usepackage{subcaption}
\usepackage{cancel}

\DeclareMathAlphabet{\pazocal}{OMS}{zplm}{m}{n}

\usepackage[ruled,vlined]{algorithm2e}

\usepackage[mathscr]{euscript}
\usepackage{cite}
\usepackage{psfrag}
\ifCLASSINFOpdf
\usepackage[pdftex]{graphicx}
\graphicspath{{../pdf/}{../jpeg/}}
\DeclareGraphicsExtensions{.pdf,.jpeg,.png}

\else
\usepackage[dvips]{graphicx}
\graphicspath{{../eps/}}
\DeclareGraphicsExtensions{.eps}
\fi
\usepackage{amsmath}
\usepackage{array}

\usepackage{tikz}
\usepackage[utf8]{inputenc}
\usepackage{pgfplots} 
\usepackage{pgfgantt}
\usepackage{pdflscape}
\pgfplotsset{compat=newest} 
\pgfplotsset{plot coordinates/math parser=false}
\usepackage{grffile}
\usetikzlibrary{plotmarks}
\usepackage{amsthm}
\usepackage{amsmath}%
\usepackage{MnSymbol}%
\usepackage{wasysym}
\usepackage{subcaption}
\usepackage[mathscr]{euscript}
\usepackage{cite}
\usepackage{amsfonts}
\usepackage{pgf}
\usepackage{mathtools}
\usepackage{graphicx}
\usepackage{array}
\usepackage{tikz}
\usepackage[utf8]{inputenc}
\usepackage{pgfplots} 
\usepackage{pgfgantt}
\usepackage{pdflscape}
\pgfplotsset{compat=newest} 
\pgfplotsset{plot coordinates/math parser=false}
\usepackage{grffile}
\usetikzlibrary{plotmarks}
\usepackage{tikzscale}
\usepackage{caption}
\usepackage{subcaption}
\usepackage{amsmath}
\usepackage{array}
\usepackage{tikz}
\usepackage[utf8]{inputenc}
\usepackage{pgfplots} 
\usepackage{pgfgantt}
\usepackage{pdflscape}
\pgfplotsset{compat=newest} 
\pgfplotsset{plot coordinates/math parser=false}
\usepackage{grffile}
\usetikzlibrary{plotmarks}
\usepackage[T1]{fontenc}
\usepackage{graphicx}
\usepackage[keeplastbox]{flushend}
\usepackage{caption}
\usepackage{overpic}
\usepackage{array}
\usepackage{color}
\usepackage{url}
\usepackage{tikz}
\usepackage[utf8]{inputenc}
\usepackage{pgfplots} 
\usepackage{cite}

\newtheorem{theorem}{Theorem}

\newtheorem{corollary}{Corollary}
\newtheorem{remark}{Remark}

\def\CN{\mathcal{C}\mathcal{N}} 
\DeclareMathOperator{\sinc}{sinc}
\usepackage[colorlinks=true,bookmarks=false,citecolor=blue,urlcolor=blue]{hyperref}  

\usepackage{xcolor,cite,etoolbox}
\makeatletter 
\pretocmd\@bibitem{\color{black}\csname keycolor#1\endcsname}{}{\fail}
\newcommand\citecolor[1]{\@namedef{keycolor#1}{\color{black}}}
\makeatother
\citecolor{Hashemi2021c}
\citecolor{Atapattu2020}
\citecolor{Cui2021}
\citecolor{Lv2021SecureSurfaces}
\citecolor{Trigui2021SecrecyPhases}
\citecolor{Wei2021}
\citecolor{He2020b}
\citecolor{MPJ+16_oneStage}
\citecolor{Popovski.2018}
\citecolor{TahaEstimation}
\citecolor{emil}

\IEEEoverridecommandlockouts
\makeatletter
\def\endthebibliography{%
	\def\@noitemerr{\@latex@warning{Empty `thebibliography' environment}}%
	\endlist
}
\makeatother

\begin{document}
	\title{Average Rate and Error Probability Analysis in Short Packet Communications over RIS-aided URLLC Systems}
	
	\author{Ramin~Hashemi,~\IEEEmembership{Student Member,~IEEE,} Samad Ali, \IEEEmembership{Member,~IEEE},
	 Nurul Huda Mahmood, \IEEEmembership{Member,~IEEE}, and Matti Latva-aho, \IEEEmembership{Senior Member,~IEEE}

	\thanks{R. Hashemi,  S. Ali, NH. Mahmood, and M. Latva-aho are with the Centre for Wireless Communications (CWC), University of Oulu, 90014 Oulu, Finland. e-mails: (ramin.hashemi@oulu.fi, samad.ali@oulu.fi, nurulhuda.mahmood@oulu.fi,  matti.latva-aho@oulu.fi).
	This research has been financially supported by Academy of Finland 6Genesis Flagship (grant 318927).

	\textcolor{black}{Preliminary results of this paper were presented in \cite{Hashemi2021c}.}

		}
}
	
	{}
	
	\maketitle
	\begin{abstract}
		In this paper, the average achievable rate and error probability of a reconfigurable intelligent surface (RIS) aided systems is investigated for the finite blocklength (FBL) regime. The performance loss due to the presence of phase errors arising from limited quantization levels as well as hardware impairments at the RIS elements is also discussed. First, the composite channel containing the direct path plus the product of reflected channels through the RIS is characterized. Then, the distribution of the received signal-to-noise ratio (SNR) is matched to a Gamma random variable whose parameters depend on the total number of RIS elements, phase errors and the channels' path loss. Next, by considering the FBL regime, the achievable rate expression and error probability are identified and the corresponding average rate and average error probability are elaborated based on the proposed SNR distribution. Furthermore, the impact of the presence of phase error due to either limited quantization levels or hardware impairments on the average rate and error probability is discussed. The numerical results show that Monte Carlo simulations conform to matched Gamma distribution to received SNR for sufficiently large number of RIS elements. In addition, the system reliability indicated by the tightness of the SNR distribution increases when RIS is leveraged particularly when only the reflected channel exists. This highlights the  advantages of RIS-aided communications for ultra-reliable and low-latency systems. The difference between Shannon capacity and achievable rate in FBL regime is also discussed. Additionally, the required number of RIS elements to achieve a desired error probability in the FBL regime will be significantly reduced when the phase shifts are performed without error.
	\end{abstract}
	
	
	\IEEEpeerreviewmaketitle
	
	\section{Introduction}
	\bstctlcite{IEEEexample:BSTcontrol}
	\textcolor{black}{The fourth industrial evolution or Industry 4.0 is aimed at digitizing industrial technology towards decentralized manufacturing of products and automation of tasks with reduced human involvement in various industrial processes. Industry 4.0 \cite{Aceto2019} is powered by Industrial Internet of Things (IIoT) \cite{Sisinni2018}  which interconnects various elements like sensors and other instruments with industrial management applications through industrial control networks (ICN) \cite{Galloway2013} to enable real-time controlling of ubiquitous actuators (AC) and machines across the smart factory. To this end, traditional wired connections are being replaced with wireless networks \cite{Varghese2014} to minimize the infrastructure expenditure, and achieve higher flexibility. However, this requires guaranteeing wired connectivity performance with wireless links. \textcolor{black}{Hence, having an ultra-reliable and high-precision physical layer communication link is of paramount importance for future industrial applications to ensure end-to-end (E2E) improvement in key performance indicators (KPIs) in a cross-layer perspective and understanding the performance limits of communication systems \cite{Changyang2021_URLLC}.}}

	
     In 5G new radio (NR), ultra-reliable and low-latency communication (URLLC) \cite{3GPPTR38.913,Popovski2014} is one of the three main service categories that address the  requirements of Industry 4.0 \cite{Chen2018b}. In URLLC, the high-reliability interprets error probabilities of less than typically $10^{-5}$, and low-latency targets to 1 ms E2E delay. The main attributes of URLLC are realized e.g. by leveraging the file contents caching \cite{Bastug2014}, utilization of shorter transmission time interval (TTI) \cite{Sachs2018a}, grant free access \cite{Berardinelli2018}, and multi-connectivity. URLLC messages usually carry control information, hence the packet lengths are generally ultra-short. As a result, the blocklength of the channel is finite which necessitates a thorough analysis of achievable rate and decoding error probability as investigated in \cite{Polyanskiy2010,Yang2014}. However, the URLLC transmission demands are not met entirely by the above solutions as the main challenge to ensure high reliability is the random nature of the propagation channel mainly due to multipath fading.

    Recently, the reconfigurable intelligent surface (RIS) technology  \cite{DiRenzo2020,Wu2020b} is introduced \textcolor{black}{as a means to improve the} spectral efficiency and coverage of wireless communication systems by influencing the propagation environment. The use of RIS \textcolor{black}{brings intelligence} to the physical channel. The structure of an RIS is composed of a metasurface where a programmable controller configures and adjusts the phase and/or amplitude response of the metasurface to modify the \textcolor{black}{behaviour of the reflection of an incident wave}. The aim of this operation is that the  received \textcolor{black}{signals at a particular receiver location} are constructively added so that the system performance enhances in terms of increasing e.g. the signal-to-noise ratio (SNR). Based on leveraging passive or active elements at each phase shifter, the RISs are classified into passive and active devices, respectively. 
    Therefore, the RIS technology can be effectively utilized in URLLC short packet transmissions under finite blocklength (FBL) regime in order to improve the IIoT networks' performance in terms of enhancing the received signal quality and ensuring high reliability. In this paper, our aim is to shed some light on the average achievable rate, and error probability analysis of RIS-aided IIoT networks in FBL regime that relies on only statistical measures of channel response.

	\subsection{Related Work}
	\subsubsection{RIS Related Studies on Channel Characterization}
    A number of studies investigate either the ergodic capacity or outage probability analysis of RIS-aided systems by identifying the characteristics of the channel response and received SNR \cite{Tao2020,Badiu2020,Ding2020,VanChien2020,Ding2020a,Han2019a,Elhattab2020,Abdullah2020,Li2020,Hou2020a,DeFigueiredo2021,Atapattu2020,Cui2021,Lv2021SecureSurfaces,Trigui2021SecrecyPhases,Qian2020}. In \cite{Tao2020} the distribution of the absolute value of the composite channel containing direct link is considered to be a Gaussian random variable (RV) for large RIS elements according to the central limit theorem (CLT) and then the ergodic capacity is studied. The authors in \cite{Badiu2020} considered a phase shift error in RIS elements which is distributed as von Mises or uniform RV, \textcolor{black}{following which} the distribution of the SNR is approximated to a Gamma RV. Then, the average error probability in an infinite blocklength channel is analyzed. In \cite{Ding2020} the authors express that the probability density function (PDF) of the reflected channel response in an RIS-aided non-orthogonal multiple access (NOMA) network is a Gaussian RV for a large number of RIS elements. Then, the diversity order analysis was studied for fully constructive adding of signals in the presence of phase error at RIS. The composite channel considered in this paper does not have a direct path between the access point (AP) and the users. In \cite{VanChien2020} the distribution of SNR is approximated as a Gamma RV by employing the moment matching technique and the ergodic capacity as well as the outage probability is studied in an infinite blocklength model i.e. conventional Shannon capacity formula is considered. Note that the assumed composite channel contains the direct link plus the reflected channel from RIS with arbitrary phase shifts so that only the statistical behavior of the phase shifts was taken into account.
    In \cite{Ding2020a} the analytical PDF of the received SNR was found in terms of cascade channel characteristics in an RIS-NOMA network  and a transmission design method was presented based on spatial division multiple access (SDMA). The impact of RIS phase error is also studied in outage probability analysis. 
    In \cite{Han2019a} the upperbound of the ergodic rate is maximized under Rician fading channel between the multi-antenna AP and the user (or RIS). The optimal phase shifts are derived based on the ergodic rate depending on the phase configuration matrix. The authors in \cite{Elhattab2020} derived a closed-form approximation for the ergodic rate of cell-edge users based on Taylor series expansion of logarithmic function. The perfect phase shift assignment is assumed and the cascade channel distribution was considered as a normal RV according to CLT. In \cite{Abdullah2020} the upperbound for the ergodic capacity is derived when perfect phase shift at the RIS is performed. The upperbound is based on the Jensen's inequality for Shannon capacity formula. The authors showed that employing a decode and forward relay will enhance the upperbound capacity significantly.

    Furthermore, the study in \cite{Li2020} analyzed the impact of finite quantization levels on the performance of an RIS-aided transmission by using a tight approximation for the ergodic capacity without assuming a specific distribution for received SNR. The best case and worst case channel characteristics are formulated as a Gamma RV with separate shape and rate parameters for each case in \cite{Hou2020a}. In \cite{DeFigueiredo2021} the absolute value of the reflected channel is considered as a Gamma RV then, the ergodic rate and outage probability analysis is investigated in terms of total RIS elements and having discrete phase shifts. Note that the analysis is performed  over infinite blocklength regime without the presence of the direct channel. \textcolor{black}{Furthermore, the authors in \cite{Atapattu2020} approximated the composite channel reflected from the RIS as a Gamma RV by Kullback-Leibler (KL) divergence method. The asymptotic analysis of outage probability for reciprocal and non-reciprocal scenarios in a two-way communication system is elaborated. In \cite{Cui2021} the SNR coverage probability is studied by modeling the SNR as a Gamma RV in an infinite blocklength regime. In addition, the authors in \cite{Lv2021SecureSurfaces,Trigui2021SecrecyPhases} invoked the Gamma RV modeling for received SNR while the physical layer security and secrecy outage probability is studied.} Finally, the authors in \cite{Qian2020} considered the optimal SNR derived in \cite{Zappone2020} and then, they proposed that the SNR distribution is composed of the product of three independent Gamma RVs and sum of two scaled non-central chi-square RVs based on the eigenvalues of the channel matrices of RIS-AP and RIS-user. The authors compare the proposed analytical distributions with the case that the SNR is only approximated with one Gamma RV. The numerical results showed that there is negligible difference in considering Gamma distribution for the SNR compared with precise analytical distributions. Furthermore, to evaluate the average achievable rate, it is intractable to perform computation of the expectations concerning SNR distribution when a complex expression is taken into account. Therefore, assuming the received SNR as a Gamma RV is tractable and sufficiently accurate. \\
	
    
    \subsubsection{URLLC Studies on Average Achievable Rate and Average Error Probability}
    Several papers investigate the performance analysis of URLLC systems \cite{Ren2020b,Ren2020a,Li2018,Li2019,Tran2020,Ren2019b,Melgarejo2020,Shehab2019} in finite blocklength (FBL) channel model. In  \cite{Ren2020b} the authors proposed to employ massive multiple-input multiple-output (MIMO) systems to leverage in IIoT networks to reduce the latency. The lower bound achievable uplink ergodic rates of massive MIMO  system with finite blocklength codes is analyzed by convexifying the rate formula which holds under specific conditions. The same authors analyzed secure URLLC in IoT applications \cite{Ren2020a} and presented resource allocation problems. The authors in \cite{Li2018,Li2019} studied the analysis of ergodic achievable data rate in FBL regime. In \cite{Li2018} a MISO network is considered and the uplink channel training was studied instead of downlink training by deriving the lower bound of the ergodic rate whereas in \cite{Li2019} channel state information (CSI) is acquired by downlink channel training. Then, the optimal number of training symbols is proposed based on the average data rate expression. 
    In \cite{Tran2020} the downlink MIMO NOMA systems' average error probability under Nakagami-m fading model is investigated in FBL regime. It should be noted that in \cite{Tran2020} the ergodic capacity analysis is not addressed and the average error probability is studied based on a well-known linear approximation for the Q-function. 
    In \cite{Ren2019b} the analysis of lower bound achievable ergodic rate in URLLC transmission is presented based on convexifying the achievable rate function where the function is convex on a specific interval. Grant-free uplink access for an RIS-assisted industrial MIMO network is studied in \cite{Melgarejo2020} as well as outage probability is done based on numerical Monte Carlo simulations. The effective capacity which is the maximum transmission rate under certain delay constraints was studied in  \cite{Shehab2019}. The authors introduced a closed-form expression for the effective capacity in a Rayleigh fading channel.

	\subsection{Contributions}
    Even though the aforementioned studies cover the topics of RIS and short-packet communication, to the best of our knowledge, there is no previous reports on the the performance analysis of an RIS-aided transmission with/without the presence of phase noise in an FBL regime for URLLC applications. Motivated by the above works we aim to elaborate on the analysis of average achievable rate and block error probability of RIS-aided factory automation wireless transmissions under FBL model. We extend our results for the case \textcolor{black}{with errors in the RIS} phase shift adjustments \textcolor{black}{arising due to, e.g., limitation of quantization bits or hardware imperfections}. The contributions of our work are summarized in the following
	\begin{itemize}
	    \item The downlink received signal containing the direct link plus a reflected signal from the RIS to the AC is identified. Then, the received SNR is statistically matched to a Gamma RV with/without the presence of phase noise at the RIS elements which is due to the quantization error or hardware impairments at the RIS phase controller. 
	    \item The average achievable rate and error probability assuming FBL regime is mathematically elaborated in terms of the characteristics of Gamma RV for the SNR. 
	    \item Since the results involves computing high-complexity functions, a tractable and closed-form lower bound formula for the average achievable rate and error probability under FBL regime is presented.
	    \item The impact of quantization error and hardware impairments are modeled as a uniform RV in the phase shift argument. Then, the corresponding modifications due to these effects on the SNR distributions are studied in detail and mathematical equations for the mean and the variance of presented Gamma RV are identified. \textcolor{black}{Furthermore, a closed-form formula which is useful in design considerations to find total channel blocklength value is studied in terms of the average rate and dispersion functions as well as the effect of the phase error on total channel blocklength}.
	\end{itemize}
	\subsection{Notations and Structure of the Paper}
	In this paper, $\textbf{h} \sim \CN(\textbf{0}_{N\times 1},\textbf{C}_{N\times N})$ denotes circularly-symmetric (central) complex normal distribution vector with zero mean $\textbf{0}_{N\times 1}$ and covariance matrix $\textbf{C}$. The operators $\mathbb{E}[\cdot]$ and $\mathbb{V}[\cdot]$ denote the statistical expectation and variance, respectively. Also, $X\sim \Gamma(a,b)$ denotes gamma random variable with shape and rate  parameters $a$ and $b$, respectively. \textcolor{black}{A uniform distributed random variable with range $[a,b]$ is shown as $Y\sim\mathcal{U}(a,b)$.} The operation $[\cdot]^H$ denotes the conjugate transpose of a matrix or vector.

	The structure of this paper is organized as follows. In Section II, the system model and mathematical identification of the received SNR and its distribution is presented. Section III presents the derivation of average rate and average error probability. In Section IV we extend our derivations to the case where phase error occurs. The numerical results are presented in Section V. Finally, Section VI concludes the paper.
	
	\section{System Model}
    Consider the downlink of an RIS-aided network consisting of a single antenna AP and AC where the RIS has $N = N_1 \times N_2$ elements. The channel response between the AP and AC has a direct path component plus a reflected channel from the RIS. Let us denote the direct channel as $h_{\text{AC}}^{\text{AP}} \sim \CN(0,\eta^{\text{AP} \rightarrow \text{AC}})$ where $\eta^{\text{AP} \rightarrow \text{AC}}$ denotes the path loss attenuation due to large scale fading. $\textbf{h}_{\text{RIS}}^{\text{AP}} \in \mathbb{C}^{N \times 1}$ and $\textbf{h}_{\text{AC}}^{\text{RIS}} \in \mathbb{C}^{N \times 1}$ represent the vector channels from the AP to the RIS and from the RIS to the AC, respectively. \textcolor{black}{We assume a block-fading channel model where the coefficients are quasi-static in each coherence interval. In addition, under the assumption of half-wavelength spacing between the RIS elements, the spatial correlation matrix of the RIS channel vector is very close to the independent and identically distributed (i.i.d.) Rayleigh fading in an isotropic propagation \cite{emil}. Therefore, all channel coefficients are assumed to be mutually independent and are modeled as circularly-symmetric Gaussian random variable random variables due to isotropic scattering\footnote{\textcolor{black}{It is worth noting that in practice the covariance matrices of reflected channel vectors may not be diagonal, however, for the sake of analytical tractability we assume diagonal covariance matrices and leave the more practical case as a future research topic.}}.} The channel vector  $\textbf{h}_{\text{RIS}}^{\text{AP}}$ is distributed as $\CN(\textbf{0}_{N \times 1},\boldsymbol{\eta}^{\text{AP}\rightarrow\text{RIS}}_{N \times N})$  where $\boldsymbol{\eta}^{\text{AP}\rightarrow\text{RIS}}=\text{diag}(\eta^{\text{AP}\rightarrow\text{RIS}}_1,...,\eta^{\text{AP}\rightarrow\text{RIS}}_N)$ is a diagonal matrix including the path loss coefficients from the AP to the RIS elements. Similarly, the channel between the RIS and AC is distributed as $\textbf{h}_{\text{AC}}^{\text{RIS}} \sim \CN(\textbf{0}_{N \times 1},\boldsymbol{\eta}^{\text{RIS}\rightarrow\text{AC}}_{N \times N}) $ where  $\boldsymbol{\eta}^{\text{RIS}\rightarrow\text{AC}}=\text{diag}(\eta^{\text{RIS}\rightarrow\text{AC}}_1,...,\eta^{\text{RIS}\rightarrow\text{AC}}_N)$ denotes the covariance matrix in this case. \textcolor{black}{Since, our aim is to obtain the upper bound performance of an RIS-aided URLLC systems, we assume that the direct path channel coefficients and the reflected channel vectors are perfectly available at the AP. Nevertheless, the channel estimation at the RIS-aided systems is thoroughly investigated in the literature, e.g., in \cite{He2020b,Wei2021,TahaEstimation}. For example, in \cite{He2020b} a general two-stage framework based on combined bilinear sparse matrix factorization and matrix completion is presented to extract the vector channel responses which consist of transmitter (AP) to RIS and RIS to receiver (AC). In addition, a channel estimation strategy based on compressed sensing and deep learning by activation of a number of RIS elements is proposed in \cite{TahaEstimation}.} Furthermore, in factory automation environments each actuator is almost in a fixed location and there is approximately low velocity in different modules. Therefore, the quasi-static channel fading model can be applied here. We assume that there is no interference for simplicity and leave the analysis in an interference network as future work. 
	\begin{figure}[t]
		\centering
		\includegraphics[width=0.5\textwidth]{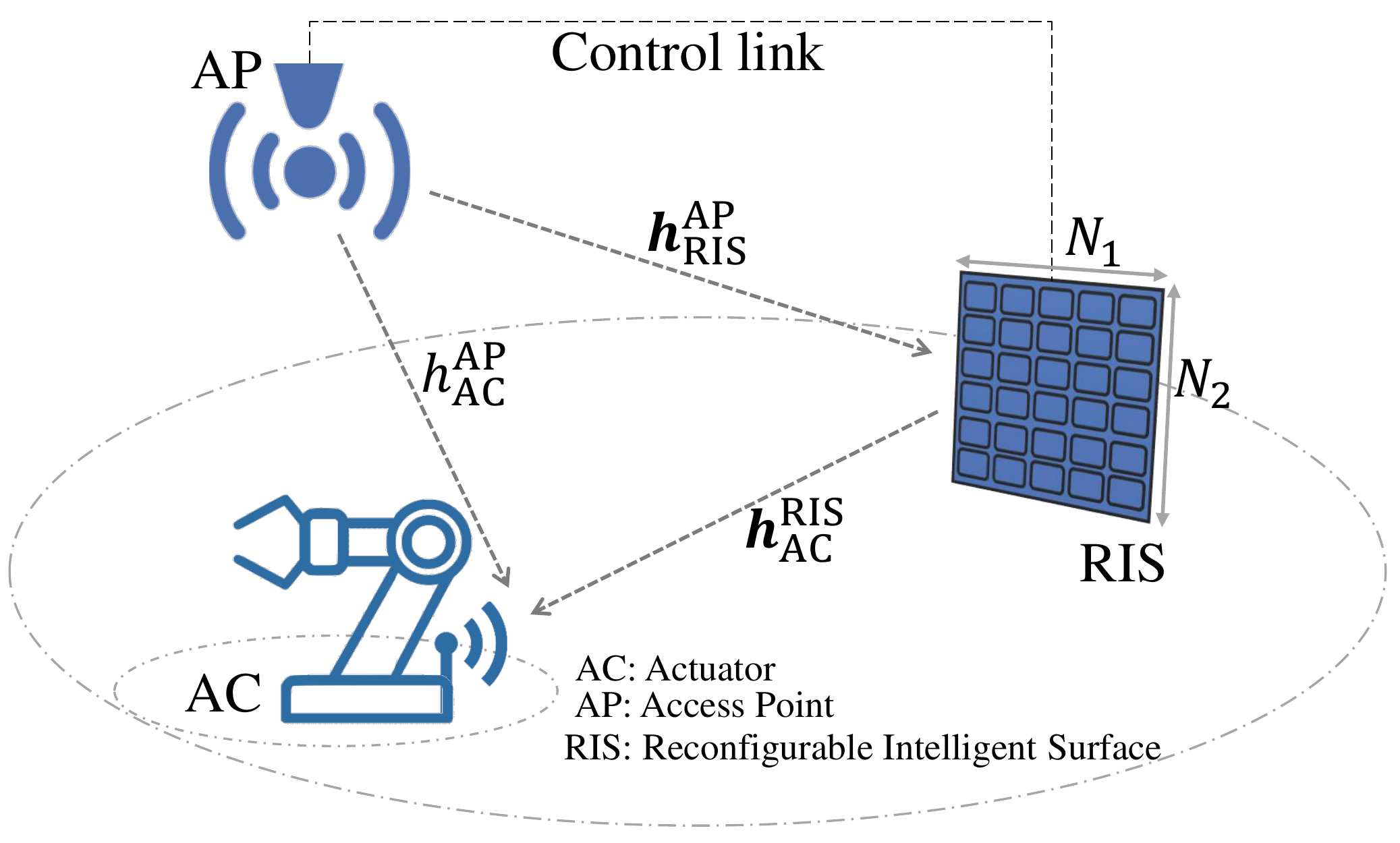}
		\caption{The system model.}
		\label{fig:2}
    \end{figure}
	
	\textcolor{black}{This work assumes a single-shot transmission~\cite{MPJ+16_oneStage, Anand2018} which means that the transmitter has to meet the target reliability/block error rate (BLER) without any retransmissions.  Such an assumption allows us to analyze the lower bound of URLLC performance since retransmissions improve reliability, albeit at the cost of additional latency~\cite{Popovski.2018}. Here, we assume that the transmission slot is at least equal to the latency budget plus propagation and processing delay.} The received signal at the AC is given by
	\begin{flalign}
	    y(t) = \left(h_{\text{AC}}^{\text{AP}} + {\textbf{h}_{\text{AC}}^{\text{RIS}}}^H\boldsymbol{\Theta}{\textbf{h}_{\text{RIS}}^{\text{AP}}}\right) s(t) +  n(t),
	\end{flalign}
	where $s(t)$ is the transmitted symbol from the AP with $\mathbb{E}[|s(t)|^2] = p$ in which $p$ is the transmit power, and $n(t)$ is the additive white Gaussian noise with $\mathbb{E}[|n(t)|^2] = N_0 W$ where $N_0$, $W$ are the noise spectral density and the system bandwidth, respectively. The complex reconfiguration matrix $\boldsymbol{\Theta}_{N \times N}$ indicates the phase shift and the amplitude attenuation of RIS which is defined as 
	\begin{flalign}
	    \boldsymbol{\Theta}_{N \times N} = & \text{diag}(\beta_1 e^{j\theta_1},\beta_2 e^{j\theta_2},...,\beta_N e^{j\theta_N}), \nonumber \\ 
	    \beta_n \in & [0,1], \quad \forall n \in \mathcal{N} \nonumber \\
	    \theta_n \in & [-\pi,\pi), \quad \forall n \in \mathcal{N}
	\end{flalign}
	where $\mathcal{N} = \{1,2,...,N\}$. Note that in our model we have assumed that the RIS elements \textcolor{black}{have no coupling or there is no joint processing among elements \cite{DiRenzo2020}}. Hence, the phase shifts and amplitude control are done independently. The phase alignment error of the RIS is defined as \textcolor{black}{$\phi_n = \angle{h^{\text{AP}}_{\text{AC}}}- \angle [\textbf{h}_{\text{AC}}^{\text{RIS}}]_n + \angle [\textbf{h}_{\text{RIS}}^{\text{AP}}]_n + \theta_n $} which occurs due to hardware limitations and/or finite number of quantization levels available at RIS phase shifters. More precisely, without considering hardware impairments the phase discrete set is selected from the following set
    \begin{flalign}
        \theta_n \in \Theta = \left\{-\pi,-\pi+\Delta,-\pi+2\Delta,...,-\pi+(Q-1)\Delta\right\}, \smallskip \forall n \in \mathcal{N}
        \label{phi_set_v1}
    \end{flalign}
    where $Q=2^b$ is the number of quantization levels, $b$ denotes the number of bits assigned to a discrete and quantized phase and $\Delta=\frac{\pi}{2^{b-1}}$ is the quantization step.

	Based on the received signal at AC and denoting the number of information bits $L$\footnote{It should be noted that $L$ that is the size of packets is assumed the same for actuator and \textcolor{black}{access point}.} that can be transmitted with target error probability $\varepsilon$ in $r$ channel uses ($r \geq 100$) the maximal achievable rate over a quasi-static additive white gaussian channel (AWGN) is given by \cite{Polyanskiy2010}
	\begin{flalign}
        R^{*}(\gamma,L,\varepsilon)=\frac{L}{r} =  \text{C}(\gamma) - Q^{-1}(\varepsilon)\sqrt{ \frac{\text{V}(\gamma)}{r}} + \mathcal{O}\left(\frac{\log_2(r)}{r}\right),
         \label{achievable_rate_urllc}
    \end{flalign}
    where $\text{C}(\gamma) = \log_2(1+\gamma)$ is the Shannon capacity formula under infinite blocklength assumption. The dispersion of the channel is defined as  $\text{V}(\gamma) = (\log_2(e))^2 \big( 1- \frac{1}{(1+\gamma)^2} \big)$. Note that $Q^{-1}(\cdot)$ is the inverse of Q-function which is defined as $Q(x) = \frac{1}{\sqrt{2\pi}}\int_{x}^{\infty}e^{-\nu^2/2}d\nu$ and
    \begin{flalign}
    \gamma = \rho\left|h_{\text{AC}}^{\text{AP}} + {\textbf{h}_{\text{AC}}^{\text{RIS}}}^H\boldsymbol{\Theta}{\textbf{h}_{\text{RIS}}^{\text{AP}}}\right|^2,
    \end{flalign}
    where $\rho = \frac{p}{N_0W}$ denotes the instantaneous SNR. The approximate of $r$ in terms of $\varepsilon$ can be expressed as \cite{Anand2018}
    \begin{flalign}
        r \approx  \frac{L}{\text{C}(\gamma)} + & \frac{\big(Q^{-1}(\varepsilon)\big)^2 \text{V}(\gamma)}{2\big(\text{C}(\gamma)\big)^2}  \\ \nonumber & + \frac{\big(Q^{-1}(\varepsilon)\big)^2 \text{V}(\gamma)}{2\big(\text{C}(\gamma)\big)^2}\sqrt{1+\frac{4L\text{C}(\gamma)}{\big(Q^{-1}(\varepsilon)\big)^2 \text{V}(\gamma)}},
    \end{flalign}
    
    Note that the term $\mathcal{O}\big(\frac{\log_2(r)}{r}\big)$ in \eqref{achievable_rate_urllc} is neglected throughout this paper as it is approximately zero for $r \geq 100$ channel uses. The decoding error probability at the AC for a packet of size $L$ transmitted via $r$ symbols is written as 
	\begin{flalign}
	    \varepsilon = Q\left(f(\gamma,r,L)\right),
	\end{flalign}
	where $ f(\gamma,r,L) = \sqrt{\frac{r}{V(\gamma)}}(\log_2(1+\gamma)-\frac{L}{r}) $. It is observed from \eqref{achievable_rate_urllc} when the blocklength approaches infinity the rate will be 
	\begin{flalign}
	    \lim_{r \rightarrow \infty} R^{*}(\gamma,L,\varepsilon) = \log_2\left(1+\rho\left|h^{\text{AP}}_{\text{AC}} + {\textbf{h}_{\text{AC}}^{\text{RIS}}}^H\boldsymbol{\Theta}{\textbf{h}_{\text{RIS}}^{\text{AP}}}\right|^2\right),
	\end{flalign}
	which is the conventional Shannon capacity formula.

	To determine the achievable average rate, we need to identify the distribution of $\gamma$. In the following, we present the related theorems and derive a closed-form and tractable approximation for the average rate.
	
	\begin{theorem}[SNR distribution]
	\label{SNR_distribution_theorem}
    Let $X = \left|h^{\text{AP}}_{\text{AC}} + {\textbf{h}_{\text{AC}}^{\text{RIS}}}^H\boldsymbol{\Theta}{\textbf{h}_{\text{RIS}}^{\text{AP}}}\right|^2$ and given $N>>1$, the distribution of $X$ is approximately matched to a Gamma random variable with the following parameters \cite{VanChien2020,Hou2020a}
    \begin{flalign}
       X \sim \Gamma(\alpha',\beta'), 
    \end{flalign}
    where $\alpha$ and $\beta$ are given in terms of first and second order moment of $X$ as 
    \begin{flalign}
    \label{alpha}
       \alpha' = \frac{(\mathbb{E}[X])^2}{\mathbb{E}[X^2]-(\mathbb{E}[X])^2},   \\  
       \beta' = \frac{\mathbb{E}[X]}{\mathbb{E}[X^2]-(\mathbb{E}[X])^2},
    \end{flalign}
    where $\mathbb{E}[X]$ and $\mathbb{E}[X^2]$ are given in \eqref{expected_value_3} and  \eqref{appndx_2}, respectively when phase error occurs ($\phi_n\neq0$, $\forall n \in \mathcal{N}$) and are given in \eqref{appndx_b_eq1} and \eqref{appndx_3} when $\phi_n=0$, $\forall n \in \mathcal{N}$. For SNR distribution we have  $\gamma = \rho X$. Therefore,   $\mathbb{E}[\gamma] = \rho \mathbb{E}[X]$ and $\mathbb{E}[\gamma^2] = \rho^2\mathbb{E}[X^2]$
    which implies that $\gamma \sim \Gamma(\alpha,\beta)$ with the same $\alpha $ as in \eqref{alpha} and $\beta = \frac{\beta'}{\rho}$.
    \end{theorem}
    \begin{proof}
        The detailed proof is given in Appendix \ref{appendixA} in the presence of phase error i.e. $\phi_n \neq 0$, $\forall n \in \mathcal{N}$ and Appendix \ref{appendixB} when $\phi_n = 0$, $\forall n \in \mathcal{N}$.
    \end{proof}
    \textcolor{black}{\begin{remark}
    For sufficiently large number of RIS elements $N$ the parameters $\alpha$ and $\beta$ in $\gamma \sim \Gamma(\alpha,\beta)$ asymptotically converge to 
    \begin{flalign}
        \lim_{N\rightarrow \infty} \alpha  & =  \lim_{N\rightarrow \infty} \frac{(\mathbb{E}[\gamma])^2}{\mathbb{V}[\gamma]} = \lim_{N\rightarrow \infty} \frac{\mathcal{O}\left(N^4\right)}{\mathcal{O}\left(N^3\right)} \rightarrow \infty, \\ 
        \lim_{N \rightarrow \infty} \beta  & =  \lim_{N\rightarrow \infty} \frac{\mathbb{E}[\gamma]}{\mathbb{V}[\gamma]} =  \lim_{N\rightarrow \infty} \frac{\mathcal{O}\left(N^2\right)}{\mathcal{O}\left(N^3\right)} \rightarrow 0.
    \end{flalign}
    where the above result is obtained regardless of with or without presence of the phase noise. In addition, it is inferred that when $N$ approaches infinity, the average value of the modeled SNR $\mathbb{E}[\gamma]=\frac{\alpha}{\beta}$ asymptotically grows with $\mathcal{O}\left(N^2\right)$. It should be noted that the asymptotic analysis, in this case, is valid when the practical value for the number of RIS elements $N$ is sufficiently large enough to ensure that the direct path loss is comparable with the reflected path coefficients.
    \label{rem1}
    \end{remark}}
    
    \section{Average Rate and Average Error Probability}
    \subsection{Average Achievable Rate}    
    In the previous section, we have modeled the SNR distribution, and the related Gamma distribution parameters $\alpha$ and $\beta$ are obtained for two cases with/without phase errors at RIS. Next, to compute the average rate, the instantaneous achievable rate should be averaged over the SNR distribution \textcolor{black}{which} we investigate in the next theorem. 
    \begin{theorem}
    \label{ergodic_rate_theo}
        The exact average achievable rate of the actuator in the RIS-aided URLLC transmission given the distribution of SNR $\gamma \sim \Gamma(\alpha,\beta)$ is expressed as 
        \begin{flalign}
           \bar{R}(L,\varepsilon) = & \mathcal{C}_1 -  \frac{Q^{-1}(\varepsilon)}{\sqrt{r}}\mathcal{C}_2   \\\nonumber
           = &\frac{\beta^{\alpha}}{\Gamma(\alpha)\ln2} \sum_{k=1}^{\infty}\frac{1}{k}\Gamma(k+\alpha) \textbf{U}(k+\alpha,1+\alpha,\beta)  \\\nonumber & - \frac{Q^{-1}(\varepsilon)\beta^{\alpha}}{\sqrt{r}\ln2}\sum_{k=0}^{\infty}\binom{\frac{1}{2}}{k}(-1)^k \textbf{U}(\alpha,1-2k+\alpha,\beta),
        \end{flalign}
    where $\mathcal{C}_1$ and $\mathcal{C}_2$ are given in \eqref{r_1_summation_form} and \eqref{r2_expectation_binomial_series3}, respectively and $\textbf{U}(a,b,z)=\frac{\int_{0}^{\infty}(1+u)^{b-a-1}u^{a-1}e^{-zu}\,du}{\Gamma(a)}$ denotes the confluent hypergeometric Kummer U function  \cite[Eq. (9.211)]{gradshteyn2014table}, and the Gamma function is denoted by $\Gamma(\alpha) = \int_0^\infty y^{\alpha-1}e^{-y}\,dy$.
    \end{theorem}
    \begin{proof}
     The instantaneous rate is given by 
    	\begin{flalign}
            R^*(\gamma,L,\varepsilon) \approx \text{C}(\gamma) - Q^{-1}(\varepsilon)\sqrt{ \frac{\text{V}(\gamma)}{r}},
             \label{achievable_rate_urllc_v2}  
        \end{flalign}
        where $\mathcal{O}\big(\frac{\log_2(r)}{r}\big)$ is ignored. To calculate the expected value of \eqref{achievable_rate_urllc_v2} in terms of the distribution of $\gamma$ we should compute the following 
        \begin{flalign}
            \bar{R}(L,\varepsilon) = \overset{\mathcal{C}_1}{\overbrace{\mathbb{E}[\log_2(1+\gamma)]}} - \frac{Q^{-1}(\varepsilon)}{\sqrt{r}}\overset{\mathcal{C}_2}{\overbrace{\mathbb{E}[\sqrt{\text{V}(\gamma)}] }} ,
            \label{expected_rate}
        \end{flalign}
        where we have used the linearity rule of the expectation. We investigate the two terms $\mathcal{C}_1$ and $\mathcal{C}_2$ involved in \eqref{expected_rate} separately. According to \eqref{expected_rate}, $\mathcal{C}_1$ is defined as 
        \begin{flalign}
            \mathcal{C}_1 = \mathbb{E}[\log_2(1+\gamma)] = \int_{0}^{\infty} \log_2(1+u)f_\gamma(u) \,du,
            \label{expectation_of_r1}
        \end{flalign}
        where $f_\gamma(u) = \frac{\beta^{\alpha}u^{\alpha-1}e^{-\beta u}}{\Gamma(\alpha)}$. Therefore, by evaluating the integral in \eqref{expectation_of_r1} and after mathematical manipulations it yields
        \begin{flalign}
            & \mathcal{C}_1 =  \mathbb{E}[\log_2 (1+\gamma)] = \int_{0}^{\infty} \log_2 (1+u)\frac{\beta^{\alpha}u^{\alpha-1}e^{-\beta u}}{\Gamma (\alpha)} \,du \nonumber \\
            & =  \frac{\beta ^{\alpha } \left(-\beta ^2\right)^{-\alpha }}{{ \Gamma (\alpha)\ln2}} \left\{  (-\beta)^{\alpha } \left[\beta  \Gamma \left(\alpha -1\right) \, _2F_2\left(1,1;2,2-\alpha ;\beta \right)+ \right. \right. \nonumber \\
            & \left. \left. \Gamma (\alpha) \left(\psi (\alpha)-\log_2 (\beta) \right) \right]+\pi  \beta ^{\alpha } \csc (\pi  \alpha) \left[\Gamma (\alpha)-\Gamma (\alpha ,-\beta) \right] \right\},
            \label{expectation_of_r1_v2}
        \end{flalign}
        where $\prescript{}{p}{\text{F}}_q(a;b;z)$ is the generalized hypergeometric function \cite[Eq. (9.1)]{gradshteyn2014table} and $\psi (\alpha ) = \frac{\Gamma^{'}(\alpha)}{\Gamma(\alpha)}$  gives the digamma function where they are standard built-in functions in most of the well-known mathematical software packages. However, we aim to write a more tractable solution for the $\mathcal{C}_1$ which involves simple computational complexity in terms of few diverse functions. To do so, first consider the following series representation for natural logarithm \cite{gradshteyn2014table}
        \begin{flalign}
            \ln(1+x) = \sum_{k=1}^{\infty}\frac{1}{k}\left(\frac{x}{x+1}\right)^k, \quad x\geq 0
            \label{natural_logarithm_series}
        \end{flalign}
        then, by substituting in \eqref{expectation_of_r1} we will have
        \begin{flalign}
            \int_{0}^{\infty}& \log_2(1+u)\frac{\beta^{\alpha}u^{\alpha-1}e^{-\beta u}}{\Gamma(\alpha)} \,du 
             \\
            \overset{a}{=} &  \frac{1}{\ln2} \sum_{k=1}^{\infty}\frac{1}{k}\int_{0}^{\infty}\left(\frac{u}{u+1}\right)^k\frac{\beta^{\alpha}u^{\alpha-1}e^{-\beta u}}{\Gamma(\alpha)} \,du \nonumber \\ \nonumber
            \overset{b}{=} & \frac{\beta^{\alpha}}{\Gamma(\alpha)\ln2} \sum_{k=1}^{\infty}\frac{1}{k}\Gamma(k+\alpha) \int_{0}^{\infty}\frac{(1+u)^{-k}u^{\alpha+k-1}e^{-\beta u}}{\Gamma(k+\alpha)}\,du,
        \end{flalign}
        where $a$ is done by exchanging the summation and integral and in $b$ a constant factor of $\Gamma(k+\alpha)$ is multiplied in numerator and denominator. We observe that the integral inside the summation of the last step is defined as the confluent hypergeometric Kummer U function  $\textbf{U}(a,b,z)$ \cite[Eq. (9.211)]{gradshteyn2014table} therefore
        \begin{flalign}
            & \mathcal{C}_1 = \frac{\beta^{\alpha}}{\Gamma(\alpha)\ln2} \sum_{k=1}^{\infty}\frac{1}{k}\Gamma(k+\alpha)\textbf{U}(k+\alpha,1+\alpha,\beta) , \label{r_1_summation_form}
        \end{flalign}
        \textcolor{black}{where the summation can be truncated to a finite number in practical numerical situations, e.g. $10^3$ with negligible mismatch}. It should be noted that \eqref{r_1_summation_form} computes only the confluent hypergeometric function and, well-known gamma function at each iteration of the summation, henceforth has less computational complexity in comparison with \eqref{expectation_of_r1_v2}.

        In the following we evaluate the expectation associated with $\mathcal{C}_2$ which is given by
        \begin{flalign}
            \mathcal{C}_2 = \mathbb{E}[\sqrt{\text{V}(\gamma)}] = \frac{1}{\ln2} \int_{0}^{\infty} \displaystyle \sqrt{1-\frac{1}{(1+u)^2}}f_\gamma(u) \,du,
            \label{expectation_of_r2}
        \end{flalign}
        to compute the above integral, we adopt the binomial expansion of the channel dispersion which is given by
        \begin{flalign}
            \sqrt{\text{V}(\gamma)} = \frac{1}{\ln2}\big( 1-\frac{1}{(1+\gamma)^2} \big)^{\frac{1}{2}} = \frac{1}{\ln2}\sum_{n=0}^{\infty}(-1)^n\binom{\frac{1}{2}}{n}(1+\gamma)^{-2n},\label{dispersion_binomial_series}
        \end{flalign}
        where $\binom{\frac{1}{2}}{n} = \frac{0.5(0.5-1)...(0.5-n+1)}{n!}$ for $n \neq 0$ and $\binom{\frac{1}{2}}{0} = 1$. Plugging \eqref{dispersion_binomial_series} in \eqref{expectation_of_r2} yields
        \begin{flalign}
            \mathcal{C}_2 = \frac{1}{\ln2}\int_{0}^{\infty} \sum_{n=0}^{\infty}(-1)^n\binom{\frac{1}{2}}{n}(1+u)^{-2n}f_\gamma(u) \,du,
            \label{r2_expectation_binomial_series}
        \end{flalign}
        where by replacing the integral and summation as well as substituting the definition of $f_\gamma(u)$, $\mathcal{C}_2$ can be rewritten as
        \begin{flalign}
            \mathcal{C}_2 = \frac{1}{\ln2}\beta^{\alpha}\sum_{n=0}^{\infty}(-1)^n\binom{\frac{1}{2}}{n}\int_{0}^{\infty}\frac{1}{\Gamma(\alpha)}(1+u)^{-2n} u^{\alpha-1}e^{-\beta u} \,du,
            \label{r2_expectation_binomial_series2}
        \end{flalign}
        we observe that the integral inside the summation can be reformulated in terms of confluent hypergeometric Kummer U function  $\textbf{U}(a,b,z)$ defined earlier, henceforth 
        \begin{flalign}
            \mathcal{C}_2 = \frac{1}{\ln2}\beta^{\alpha}\sum_{n=0}^{\infty}\binom{\frac{1}{2}}{n}(-1)^n\textbf{U}(\alpha,1-2n+\alpha,\beta).
            \label{r2_expectation_binomial_series3}
        \end{flalign}
        consequently, if we substitute the mathematical expressions obtained in \eqref{r_1_summation_form} and  \eqref{r2_expectation_binomial_series3} in the average rate formula in \eqref{expected_rate}, the final result will be obtained which completes the proof.
    \end{proof}
    
    It is worth mentioning that the average rate given in Theorem \ref{ergodic_rate_theo} involves evaluating high-computational complexity functions as well as infinite summations which may not be beneficial in practical situations and resource allocation algorithms. Therefore, we propose a tractable lower bound approximation for the average rate in the following corollary. 
    
    \begin{corollary}
        A tractable and closed-form approximate lowerbound expression for the average rate $\bar{R}(L,\varepsilon)$ is given by
        \begin{flalign}
            \bar{R}_{\text{LB}}(L,\varepsilon) \approx \Tilde{\mathcal{C}}_1 & - \frac{Q^{-1}(\varepsilon)}{\sqrt{r}}\Tilde{\mathcal{C}}_2 \label{Corr_1}  = \log_2\left(1+\frac{\alpha^2}{\beta(\alpha+1)}\right) \\ \nonumber & - \frac{Q^{-1}(\varepsilon)}{2\ln2\sqrt{r}} \left(2-\beta +e^{\beta } \beta  (\alpha +\beta -1) \text{E}_{\alpha }(\beta )\right),
        \end{flalign}
        where $\text{E}_n(z)$ is the exponential integral function \cite[Eq. (8.211)]{gradshteyn2014table},  $\Tilde{\mathcal{C}}_1$ and $\Tilde{\mathcal{C}}_2$ are given in \eqref{Jensen_v2} and \eqref{expectation_of_r2_approximate_dispersion},  respectively.
    \end{corollary}
    \begin{proof}
    First, we study the term involving Shannon capacity, i.e. $\mathcal{C}_1=\mathbb{E}[\log_2(1+\gamma)]$. Invoking  Jensen's inequality \cite{Larsson2015} we have
    \begin{flalign}
        \mathcal{C}_1 = \mathbb{E}[\log_2(1+\gamma)] \geq \log_2(1+\frac{1}{\mathbb{E}[\gamma^{-1}]}),
        \label{Jensen}
    \end{flalign}
    where to arrive in a lower-bound we should elaborate the right-hand of \eqref{Jensen}. To continue, we apply Taylor series expansion of $\frac{1}{\gamma}$ then, take average from both sides which yields
    \begin{flalign}
        \mathbb{E}[\gamma^{-1}] \approx \frac{1}{\mathbb{E}[\gamma]} + \frac{\mathbb{V}[\gamma]}{\mathbb{E}[\gamma]^3} = \frac{\mathbb{E}[\gamma^2]}{\mathbb{E}[\gamma]^3},
        \label{taylor_series}
    \end{flalign}
    substituting \eqref{taylor_series} in \eqref{Jensen} yields
    \begin{flalign}
        \mathcal{C}_1  \geq \tilde{\mathcal{C}}_1 =  \log_2(1+\frac{\mathbb{E}[\gamma]^3}{\mathbb{E}[\gamma^2]}) =\log_2\left(1+\frac{\alpha^2}{\beta(\alpha+1)}\right),
        \label{Jensen_v2}
    \end{flalign}
    where $\mathbb{E}[\gamma]^3$ and $\mathbb{E}[\gamma^2]$ can be evaluated straightforward from the SNR distribution $\gamma \sim \Gamma(\alpha,\beta)$ such that $\mathbb{E}[\gamma] = \frac{\alpha}{\beta}$ and $\mathbb{E}[\gamma^2] = \frac{\alpha(\alpha+1)}{\beta^2}$ where $\alpha$ and $\beta$ are investigated in Theorem \ref{SNR_distribution_theorem}. It is worth mentioning that another approximation for $\mathcal{C}_1$ can be written based on truncated Taylor series expansion of $\ln(1+\gamma)$ as given by
    \begin{flalign}
        \ln(1+\gamma) \approx \ln(1+\gamma_0)+\frac{1}{1+\gamma_0}(\gamma-\gamma_0)-\frac{1}{2(1+\gamma_0)^2}(\gamma-\gamma_0)^2,
    \end{flalign}
    then, by replacing $\gamma_0=\mathbb{E}[\gamma]$ and taking average from both sides which yields
    \begin{flalign}
        \mathcal{C}_1 = \mathbb{E}[\log_2(1+\gamma)] \approx  \log_2(1+\mathbb{E}[\gamma])-\frac{\mathbb{V}[\gamma]}{2(1+\mathbb{E}[\gamma])\ln2},
        \label{another_approx_r2}
    \end{flalign}
    however, using the above expression does not give either lowerbound or upperbound for $\mathcal{C}_1$ therefore, for comparison purposes the lowerbound given in \eqref{Jensen_v2} is advantageous. \\

    Next, we investigate $\mathcal{C}_2$ in \eqref{expectation_of_r2}. To do so, we apply the truncated binomial approximation $(1+x)^\vartheta \approx 1+\vartheta x$ for $|x| < 1$, $|\vartheta x| < 1$ to approximate the channel dispersion as $\sqrt{\text{V}(u)}=\frac{1}{\ln2}\sqrt{1-\frac{1}{(1+u)^2}} \approx \frac{1}{\ln2}(1-\frac{1}{2(1+u)^2})$. 
    By substituting the approximate dispersion expression in \eqref{expectation_of_r2} we will have
    \begin{flalign}
        \tilde{\mathcal{C}}_2 = & \mathbb{E}[\sqrt{\text{V}(\gamma)}] \approx \frac{1}{\ln2} \int_{0}^{\infty} \displaystyle (1-\frac{1}{2(1+u)^2})f_\gamma(u) \,du \nonumber \\  & =\frac{1}{\ln2} \int_{0}^{\infty} \displaystyle (1-\frac{1}{2(1+u)^2})\frac{\beta^\alpha u^{\alpha-1}e^{-\beta u}}{\Gamma(\alpha)} \,du  \nonumber \\  
        & = \frac{1}{\ln2}\left(1-\frac{\beta^\alpha}{2}\int_{0}^{\infty} \displaystyle \frac{1}{(1+u)^2}\frac{u^{\alpha-1}e^{-\beta u}}{\Gamma(\alpha)} \,du\right) \nonumber \\  
        &\overset{a}{=} \frac{1}{2\ln2} \left(2-\beta +e^{\beta } \beta  (\alpha +\beta -1) \text{E}_{\alpha }(\beta )\right).
            \label{expectation_of_r2_approximate_dispersion}
    \end{flalign}
    where $a$ is obtained through integrating by part and some manipulations. Finally, by replacing \eqref{Jensen_v2} and \eqref{expectation_of_r2_approximate_dispersion} in \eqref{expected_rate} the approximate result will be obtained. Note that since $\mathcal{C}_1$ is lower bounded by $\tilde{\mathcal{C}}_1$ and $\mathcal{C}_2 \leq \tilde{\mathcal{C}}_2$ the achievable rate will attain its lower bound because of negative sign in $\mathcal{C}_2$ in \eqref{Corr_1}. 
    \end{proof}
    
    \subsection{Average Decoding Error Probability}
    In order to compute the average decoding error probability one needs to evaluate the following
    \begin{flalign}
        \bar{\varepsilon} \approx \mathbb{E}\Big[Q\big(\sqrt{\frac{r}{V(\gamma)}}(\log_2(1+\gamma)-\frac{L}{r})\big)\Big],
        \label{error_average}
    \end{flalign}
    where the expected value is taken over the distribution of $\gamma$ and, still $r$ is large enough to have a better accuracy. It should be noted that directly computing \eqref{error_average} is intractable to achieve a closed-form solution. Therefore, we adopt a linear approximation to the function inside the the expected value as
    \begin{flalign}
         Q\left(\sqrt{\frac{r}{V(\gamma)}}(\log_2(1+\gamma)-\frac{L}{r}) \right) \approx \mathsf{g}(\gamma)
    \end{flalign}
    where $\mathsf{g}(\gamma)$ is given by \cite{Makki2014}
    \begin{flalign}
        \mathsf{g}(\gamma) = 
        \begin{cases} 
          1, & \gamma \in [0,\kappa_0) \\
          \frac{1}{2}+\xi_0(\gamma- \xi_1), &  \gamma \in [\kappa_0,\kappa_1] \\
          0. & \gamma \in [\kappa_1,\infty) 
       \end{cases}
    \end{flalign}
    where $ \xi_0 = -\sqrt{\frac{r}{2\pi(2^{\frac{2L}{r}}-1)}}$, $ \xi_1=2^\frac{L}{r}-1$, $ \kappa_0=\xi_1+\frac{1}{2\xi_0}$ and $\kappa_1=\xi_1-\frac{1}{2\xi_0}$. Consequently, the average error probability defined in \eqref{error_average} will be 
    \begin{flalign}
        \bar{\varepsilon} = & \int_{0}^{\infty} Q\left(\sqrt{\frac{r}{V(u)}}(\log_2(1+u)-\frac{L}{r})\right) f_\gamma(u)\,du 
        \nonumber \\  
        & \approx \int_{0}^{\kappa_0} f_\gamma(u)\,du + \int^{\kappa_1}_{\kappa_0} \left(\frac{1}{2}+\xi_0(u- \xi_1)\right)f_\gamma(u)\,du \nonumber \\
        = &  (\frac{1}{2}+\xi_0\xi_1)F_\gamma(\kappa_0) + (\frac{1}{2}-\xi_0\xi_1)F_\gamma(\kappa_1) + \xi_0 \int^{\kappa_1}_{\kappa_0} u f_\gamma(u)\,du \nonumber  \\ \nonumber 
        = & \frac{\xi_0}{\beta  \Gamma (\alpha )} \bigg( \beta (\kappa_0 - \kappa_1)  \Gamma (\alpha )-\beta \kappa_0 \Gamma (\alpha ,\kappa_0 \beta )+\beta  \kappa_1 \Gamma (\alpha ,\kappa_1 \beta ) \nonumber \\  & \quad \quad \quad \quad \quad  +\Gamma (\alpha +1,\kappa_0 \beta )-\Gamma (\alpha +1,\kappa_1 \beta )  \bigg),
        \label{epsilon_average}
    \end{flalign}
    where $F_\gamma(\gamma)=\frac{\gamma(\alpha,\beta)}{\Gamma(\alpha)}$ is the cumulative distribution function of the SNR that is investigated in Theorem \ref{SNR_distribution_theorem} and $\gamma(a,b)$, $\Gamma(a,b)$ are the lower and upper incomplete gamma functions, respectively \cite{gradshteyn2014table}. Consequently, \eqref{epsilon_average} gives a closed form formula for the average error probability. 
    
    \subsection{The Required Number of Channel Blocklengths}
    \textcolor{black}{By solving a quadratic equation in terms of $\varpi = {\sqrt{r}}$ in \eqref{achievable_rate_urllc} and ignoring the higher order terms $\mathcal{O}\big(\frac{\log_2(r)}{r}\big)$, an estimate of the average number of channel blocklengths $\bar{r}=\mathbb{E}[r]$ in terms of number of RIS elements $N$ and phase shifts will be obtained as follows
    \begin{flalign}
        \bar{r} & \approx \label{estimate_r} \frac{L}{\mathcal{C}_1} + \frac{\left(Q^{-1}(\varepsilon)\mathcal{C}_2\right)^2}{2\mathcal{C}_1^2}+\frac{Q^{-1}(\varepsilon)\mathcal{C}_2\sqrt{\left(Q^{-1}(\varepsilon)\mathcal{C}_2\right)^2+4\mathcal{C}_1L}}{4\mathcal{C}_1^2}, 
    \end{flalign}
    therefore, $\bar{r}$ can be evaluated straightforwardly after determining $\mathcal{C}_1$ and  $\mathcal{C}_2$ that were investigated in earlier sections. 
    \begin{remark}
    \label{remark_channel_use}
        The average channel blocklength $\bar{r}$ is a decreasing function in terms of the total number of RIS elements $N$.
    \end{remark}
    \begin{proof}
    Please refer to Appendix \ref{appendixC}.
    \end{proof}
    Consequently, according to Remark \ref{remark_channel_use}, the average channel use $\bar{r}$ in \eqref{estimate_r} reduces with respect to increase in the RIS elements $N$ that results in less required channel blocklengths for transmission.} In other words, given the channel blocklength which is defined as $r=TW$ (which $T$ is the transmission duration and $W$ denotes the available bandwidth) we can expect that when total the number of RIS elements increases the required blocklength to transmit the symbols will be reduced. \textcolor{black}{This results in lower transmission time as well as less power consumption at the AP which highlights the suitability of RIS at URLLC systems in the FBL regime.}


    \section{Impact of Phase Error}
    In this section we investigate the performance loss due to the presence of phase error in RIS elements which \textcolor{black}{may occur due to quantization of the RIS elements' phase according to the number of bits assigned to each discrete phase by the controller} or hardware impairments \cite{Wu2020b,DiRenzo2020}. \textcolor{black}{Considering either hardware impairments or quantization error the resultant impact will be on the distribution of $\phi_n$, $\forall n \in \mathscr{N}$. Therefore, the SNR parameters namely $\alpha$ and $\beta$ are needed to recompute. Henceforth, we propose a general framework to have a clear understanding of this effect. Different distributions can be analyzed as straightforward without loss of generality.}

    Let us assume that there are only $Q$ quantization levels in which $b = \log_2(Q)$ bits are assigned to each discrete phase shift. The RIS chooses each phase shift from the following set
    \begin{flalign}
        \theta_n \in \Theta = \big\{-\pi,-\pi+\Delta,-\pi+2\Delta,...,-\pi+(Q-1)\Delta\big\}, \smallskip \forall n \in \mathcal{N}
        \label{phi_set}
    \end{flalign}
    where $\Delta = \frac{\pi}{2^{b-1}}$. It should be noted that a linear quantizer has an error $e$ which \textcolor{black}{spans} uniformly over $\frac{-\Delta}{2}\leq e\leq \frac{\Delta}{2}$. Let us denote the phase error as $\phi_n \sim \mathcal{U}(-\epsilon\pi,\epsilon \pi)$, $\forall n \in \mathcal{N}$ where $\epsilon = \frac{1}{2^b} \in (0,1]$. Additionally, the phase errors of each element are independently distributed. In this case, we should compute the expected values of \eqref{expected_value_3} and \eqref{appndx_2} in terms of the phase error distribution. Therefore, the following expressions are needed for $\forall n\neq m\neq n'$
    \begin{subequations}
    \begin{flalign}
        \label{eq_start}
        &\mathbb{E} [\cos(\phi_n)] = \sinc(\epsilon),\\
        &\mathbb{E} [\cos^2(\phi_n)] =\frac{1}{2} + \frac{\sinc(2\epsilon)}{2}, \\
        &\mathbb{E} [\cos(\phi_n-\phi_m)] = \sinc^2(\epsilon),  \\
        &\mathbb{E} [\cos^2(\phi_n-\phi_m)] = \frac{1}{2}+\frac{\sinc^2(2\epsilon)}{2},  \\
        &\mathbb{E} [\cos(\phi_n)\cos(\phi_n-\phi_m)] = \sinc(\epsilon)(\frac{1}{2}+\frac{\sinc(2\epsilon)}{2}),  \\
        &\mathbb{E} [\cos(\phi_n-\phi_m)\cos(\phi_{n'}-\phi_n)] = \frac{1}{2}(1+\sinc(2\epsilon))\sinc^2(\epsilon),  \label{eq_end}
    \end{flalign}
    \end{subequations}
    consequently, the average rate and error probability will be achieved after substituting \eqref{eq_start}--\eqref{eq_end} in \eqref{expected_value_3} and \eqref{appndx_2} which yields
    \begin{flalign}
        \mathbb{E}[\gamma] = & \rho\big(\eta^{\text{AP}\rightarrow \text{AC}} + N \eta^{\text{AP}\rightarrow\text{RIS}} \eta^{\text{RIS}\rightarrow\text{AC}} \nonumber \\ & + \frac{\pi^2\sinc^2(\epsilon)N(N-1)}{16}\eta^{\text{AP} \rightarrow\text{RIS}} \eta^{\text{RIS}\rightarrow\text{AC}}  \nonumber \\ & +\frac{N\pi\sinc(\epsilon)}{4}\sqrt{\pi\eta^{\text{AP}\rightarrow \text{AC}}\eta^{\text{AP} \rightarrow\text{RIS}} \eta^{\text{RIS}\rightarrow\text{AC}}}  \big),
        \label{E_gamma_phase_error}
    \end{flalign}
    \begin{figure*}
    \begin{flalign}
        \mathbb{E}[\gamma^2] = 
        \rho^2 \bigg[& 2\varsigma^2 + N\varsigma\varrho\vartheta\Big(4 + 2\sinc(2\epsilon)  + \frac{3\pi^2(N-1)}{8}\sinc^2(\epsilon) \Big)  +   \sqrt{\varsigma^3\varrho\vartheta} \Big( \frac{3\pi\sqrt{\pi}N}{4}\sinc(\epsilon) \Big) 
        \nonumber \\ & +N\varrho^2\vartheta^2 \bigg(N+3+\frac{\pi^2}{16}(N-1)(2N+5)\sinc^2(\epsilon) + (N-1)\sinc^2(2\epsilon) \nonumber \\ 
        & + \frac{\pi^2}{8} (N-1)(N-2)\sinc^2(\epsilon)(1+\sinc(2\epsilon))  +\frac{\pi^4}{2^8} \sinc^4(\epsilon)(N-1)(N-2)(N-3) \bigg)
        \nonumber \\ 
        &
        + N\sqrt{\varsigma\varrho^3\vartheta^3}\sinc(\epsilon) \pi \sqrt{\pi}\Big( \frac{8N+1}{8} +  \frac{(N-1)}{2}\big(\sinc(2\epsilon) +\frac{(N-2)\pi^2}{16}\sinc^2(\epsilon)\big) \Big) \bigg] ,
        \label{uniform_phase_error_EX2}
    \end{flalign}
    \hrule
    \end{figure*}
    and $\mathbb{E}[\gamma^2]$ is given in \eqref{uniform_phase_error_EX2} on top of the next page. Note that as in Appendix \ref{appendixA} we have $\varsigma = \eta^{\text{AP} \rightarrow \text{AC}}$, $\varrho = \eta^{\text{AP} \rightarrow \text{RIS}}$ and $\vartheta = \eta^{\text{RIS} \rightarrow \text{AC}}$. By noting that $\gamma \sim \Gamma(\alpha,\beta)$ where $\alpha=\frac{(\mathbb{E}[\gamma])^2}{\mathbb{E}[\gamma^2]-(\mathbb{E}[\gamma])^2}$, $\beta=\frac{\mathbb{E}[\gamma]}{\mathbb{E}[\gamma^2]-(\mathbb{E}[\gamma])^2}$ the average rate and error probability can be evaluated simply using Theorem \ref{ergodic_rate_theo} and \eqref{epsilon_average}, respectively.
    
    \textcolor{black}{\begin{remark}
    For sufficiently large number of RIS elements $N >>1$, the mean value and variance of $\gamma$ will be approximated as
    \begin{flalign}
        \mathbb{E}[\gamma]&_{N >> 1}  \approx \rho\varsigma+\frac{\rho}{16} N^2 \varrho\vartheta \pi^2\sinc^2(\epsilon), \nonumber \\ 
        \mathbb{V}[\gamma]&_{N >> 1} \approx   \nonumber \\ \nonumber &    64 N^3  \rho^2 \varrho^2\vartheta^2 \pi^2 \sinc^2(\epsilon) \big(8+\sinc(2\epsilon)+\pi^2\sinc^2(\epsilon)\big) \\ \nonumber & +  \rho^2\varsigma^2 + \frac{1}{4}N\pi^{1.5}\varsigma^{1.5}(\varrho\vartheta)^{1.5}\sinc(\epsilon) \nonumber \\ & + \mathcal{O}\left(N^{n_0}\varsigma^{n_1}\varrho^{n_2}\vartheta^{n_3}\right). \nonumber 
    \end{flalign}
    where $n_1+n_2+n_3 > 3, \enskip n_0 \in \{1,2\}$ indicate the power exponents of non-dominant terms. As a special case we have $\mathbb{V}[\gamma]_{\text{no direct channel}} << \mathbb{V}[\gamma]_{\text{with direct channel}}$. Since the direct channel path loss is slightly lower than the cascaded channel path loss, i.e.,  $\varsigma >>\varrho \vartheta$.
    \end{remark}}

    \section{Performance Evaluation}
	In what follows, we evaluate the proposed derivations and mathematical expressions numerically. Table \ref{table2} shows the chosen default values for the network parameters and geometry. Since the carrier frequency is 1900 MHz a typical antenna size will be around 15 cm. Therefore, the far-field assumption holds true in all simulation scenarios for the given network geometry.   
	\begin{table}[t]
		\caption{Simulation parameters.}
		\centering
		\begin{tabular}{ l  l }
			\hline
			Parameter & Default value \\ \hline
			RIS location in 2D plane & ($d$,10) m, $d \in [5,95]$ \\ AP location & (0,0) m \\
			AC position & \textcolor{black}{(100,0) m} \\ AP transmit power ($p$)  & 200 mW  \\ 
			Receiver noise figure (NF)  & 3  dB \\ \textcolor{black}{Target error probability ($\varepsilon$)}  & $10^{-9}$ \\
			$\phi_n, \forall n \in \mathcal{N}$ & $\sim \mathcal{U}[-\frac{\pi}{4},\frac{\pi}{4}]$ \\ Noise power density ($N_0$) & -174 dBm/Hz \\ 
			Channel blocklength ($r$) & $300$ \\ The size of packets ($L$) & $240$ bits ($\frac{L}{r}=0.8$)   \\ 
			Number of realizations & $10^4$ \\ Bandwidth ($W$) &  200 kHz \\ 
			Carrier frequency & 1900 MHz \\ AP height &  12.5 m \\ 
			\shortstack[l]{Path loss model \\ ($D$: distance in meter)} & $\text{PL(dB)}=34.53+38\log_{10}(D)$ \\ \hline
		\end{tabular}
		\label{table2}
	\end{table}
	In Fig. \ref{fig:SNR_CDFa} and Fig. \ref{fig:SNR_CDFb} the cumulative distribution function \textcolor{black}{(CDF)} of the received SNR is illustrated for two cases namely as in the presence of the direct channel between the AP and the AC and the case where there is no direct link. Besides, the results are shown  when perfect phase alignment is performed at the RIS which is referred to as $\phi_n=0, \forall n \in \mathcal{N}$ and comparison is made with uniform distributed phase noise at the RIS as a benchmark. As can be observed from Fig. \ref{fig:SNR_CDFa}, Monte Carlo simulations conform to the matched Gamma distribution for SNR. Furthermore, the impact of quantizer noise is shown in Fig. \ref{fig:SNR_CDFa}. Although there is a performance gap between leveraging $n$-bit quantizer ($n=1,2,3$) and the optimal case, \textcolor{black}{a 2-bit quantizer provides a good trade-off between performance and training overhead as the gap is $0.9$ dB (compared to $3.9$ dB gap with 1-bit) \cite{Wu2020c}. The 3-bit quantizer is more closer to the perfect case with small SNR degradation. It can be inferred that utilizing a lower-bit quantizer leads to reduction in RIS expenditure since higher bits means higher capacity requirement of the control signalling exchange for phase adjustment of the RIS elements, and hence higher control overhead. Henceforth, acquiring lower bit quantizers is mandatory.}

	In Fig. \ref{fig:SNR_CDFb} a comparison is shown between three scenarios, \textcolor{black}{I) no direct link and the RIS acts as a simple reflector, II) no direct link with phase adjusted RIS and III) with direct link and phase adjusted RIS}. As we observe, the only case which has a sharp slope and higher tightness is scenario II. It means that the values of SNR that we expect to receive in scenario II are almost in the same range whereas the range of fluctuations in cases I and III which are denoted as intervals of $*$ and $**$, respectively, is noticeable that results in eroding the perception of \textit{reliability}. \textcolor{black}{Therefore, though the SNR value is higher with the direct link, the variation in the SNR implies that the channel is less deterministic compared to case II. Hence the RIS is beneficial in guaranteeing a high reliability even if the direct link is absent (e.g., blocked). Furthermore, we observe that in the presence of the direct channel there is negligible difference between leveraging a 1-bit quantizer with a 2-bit in Fig. \ref{fig:SNR_CDFb}. This is because the direct channel is several orders of magnitude stronger than the reflected channel so that the effect of reflected channel is not dominant. }
	\begin{figure*}[t]
		\centering
		\begin{subfigure}[b]{0.49\textwidth}  
			\centering
			\includegraphics[trim = 8.5cm 8.5cm 8.5cm 8.5cm,scale=0.6]{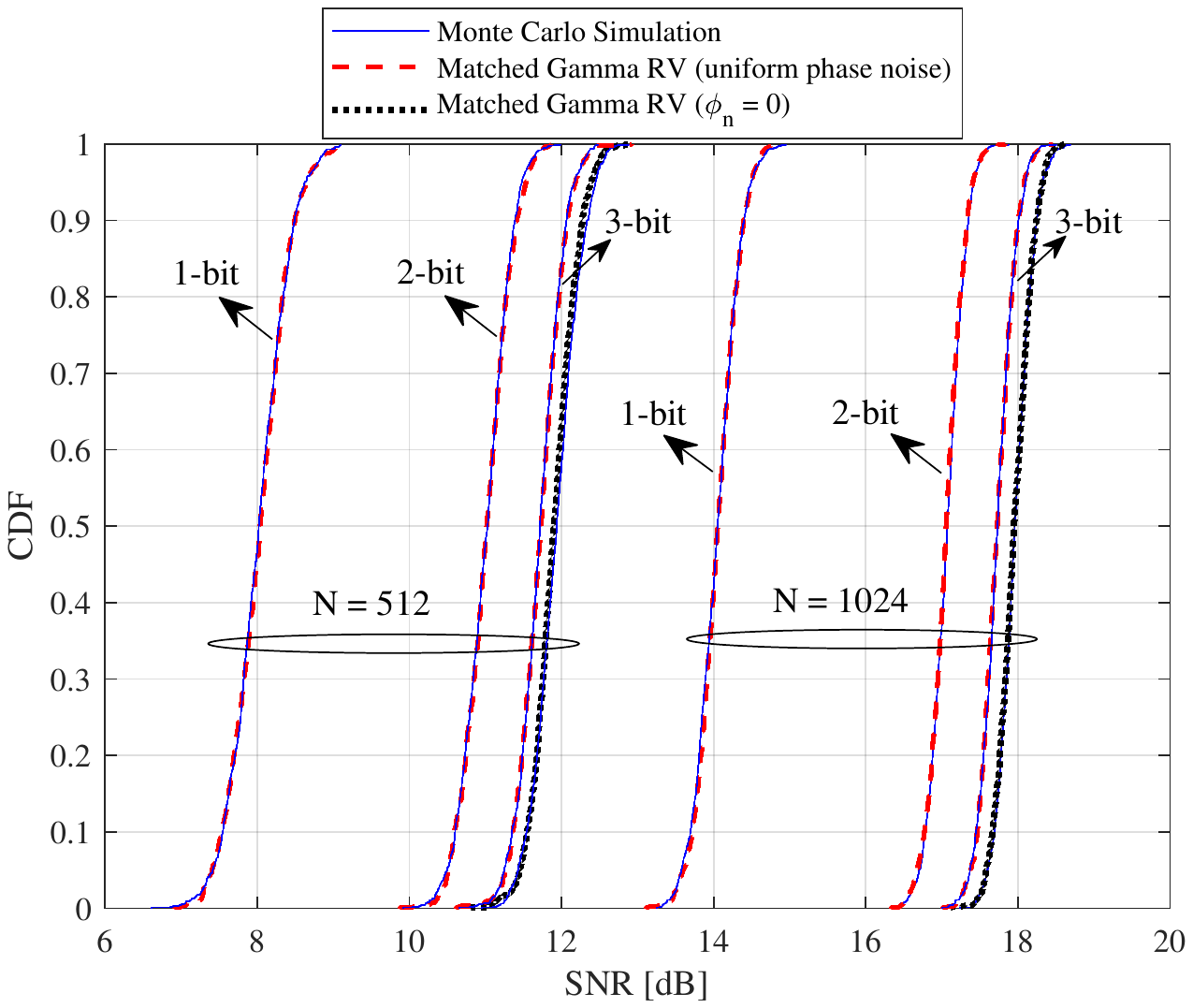}
			\caption{SNR CDFs without direct link.}
			\label{fig:SNR_CDFa}
		\end{subfigure}
		\begin{subfigure}[b]{0.5\textwidth}
			\centering    
			\includegraphics[trim = 8.5cm 8.5cm 8.5cm 8.5cm,scale=0.6]{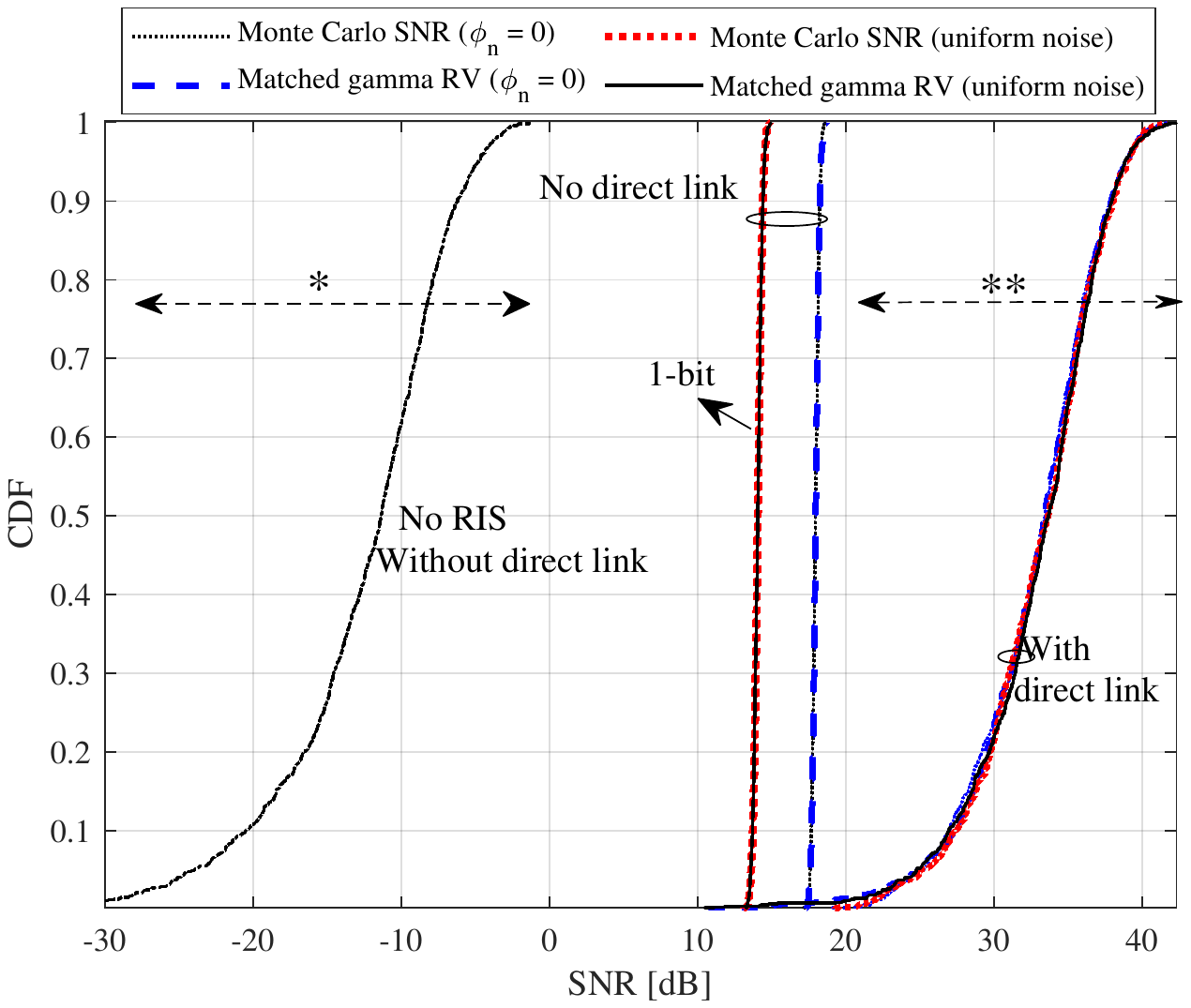}
			\caption{Comparison of SNR CDFs ($N = 1024$).}
			\label{fig:SNR_CDFb}
		\end{subfigure}
		\caption{The CDF curves of SNR and illustration of matched Gamma RV with Monte Carlo simulations.}
		\label{fig:SNR_CDF}
	\end{figure*}

    By evaluating the average achievable rate as well as taking into account the channel dispersion, a 2-bit quantizer is compared with a 1-bit quantizer in Fig. \ref{fig:ErgodicRate} along with a optimal phase setting. The results also confirm that Monte Carlo simulations approximate to the derived analytical expressions for the average rate in Theorem \ref{ergodic_rate_theo}. As it is observed, when the number of bits assigned to each discrete phase at the RIS increments, the average rate curve is very close to the optimal phase alignment case. Once more, this shows that to achieve satisfactory accuracy, a few numbers of available bits will be sufficient  instead of high precision and high complexity quantizers. Moreover, it is proved that to reach full diversity in RIS-aided communications the number of quantization levels must be $Q\geq3$ \cite{Xu2020b} that holds when $b\geq2$. Furthermore, the Shannon capacity and the gap with FBL regime is illustrated in Fig. \ref{fig:ErgodicRate}. We observe that the gap is increased until some saturation value. This is because of asymptotically converging the channel dispersion to its upperbound when the number of RIS elements increases which results in improving the SNR in other words $\lim_{\gamma \rightarrow \infty}\text{V}(\gamma) = (\log_2(e))^2$.
   	\begin{figure}[t]
		\centering
		\includegraphics[trim = 8.5cm 8.5cm 8.5cm 8.5cm,scale=0.6]{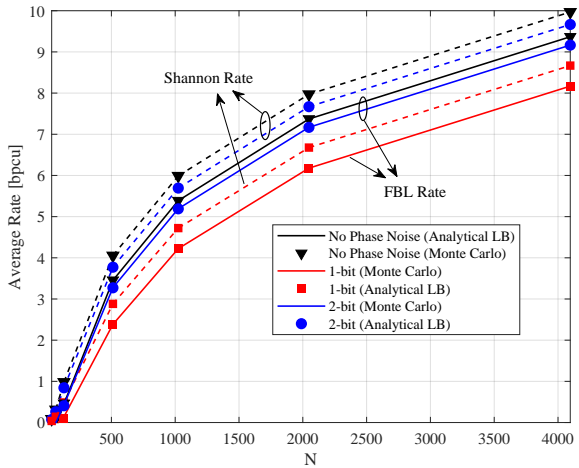}
		\caption{The average rate in terms of changing the total RIS elements without direct channel.}
		\label{fig:ErgodicRate}
    \end{figure}
    
    Next, we show the results for average error probability given in Fig. \ref{fig:ErrorProbability} when there is no direct channel. As we observe, to achieve a desired error probability the required number of RIS elements is much higher when there is phase error at the RIS compared with optimal phase setting. This shows the importance of phase alignment precision in the RIS particularly in URLLC applications. For instance, when we desire to reach an error probability of $10^{-9}$, the number of RIS elements satisfying this condition should be at least 190 elements in a perfect phase alignment scenario whereas in case of having phase errors the required number of RIS elements should be at least 290 elements in a 1-bit quantizer \textcolor{black}{and about 200 with 2-bit quantizer}. On the other hand, one may interpret these curves as a baseline for design considerations of how many RIS elements should be installed to reach sufficiently low bit error rate.  
   	\begin{figure}[t]
		\centering
		\includegraphics[trim = 9.5cm 9.5cm 9.5cm 9.5cm,scale=0.62]{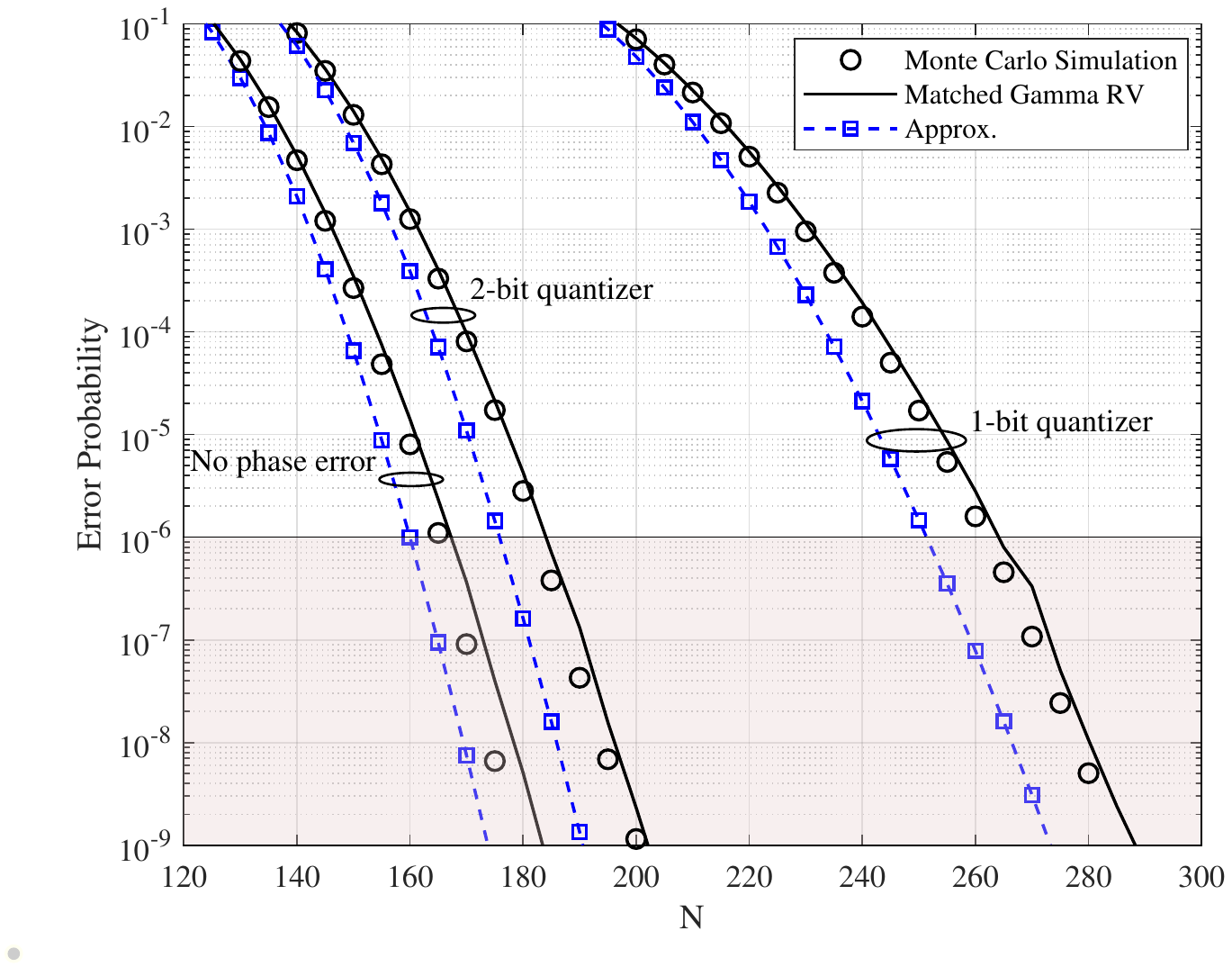}
		\caption{Average error probability versus RIS elements $N$ without direct channel.}
		\label{fig:ErrorProbability}
    \end{figure}
    
    \textcolor{black}{In Fig. \ref{fig:ErrorProbability_vs_distance} the average error probability results are illustrated for $d \in [5,95]$, where RIS is located at ($d$, 10) on the 2D plane (cf. Table \ref{table2}) (close to AP $\rightarrow$ close to AC). The number of RIS elements, $N = 512$. We observe that the error probability behavior is somewhat symmetric with respect to the distance from either the AP or the AC. The error probability degrades as the RIS moves further away from either the AP or the AC, with the worst performance observed for the case when it is equidistant from both. This is because the path loss, which is proportional to the product of the RIS distances from the AP and the AC, is the highest when the two distances are equal. }

   	\begin{figure}[t]
		\centering
		\includegraphics[trim = 9.5cm 9.5cm 9.5cm 9.5cm,scale=0.62]{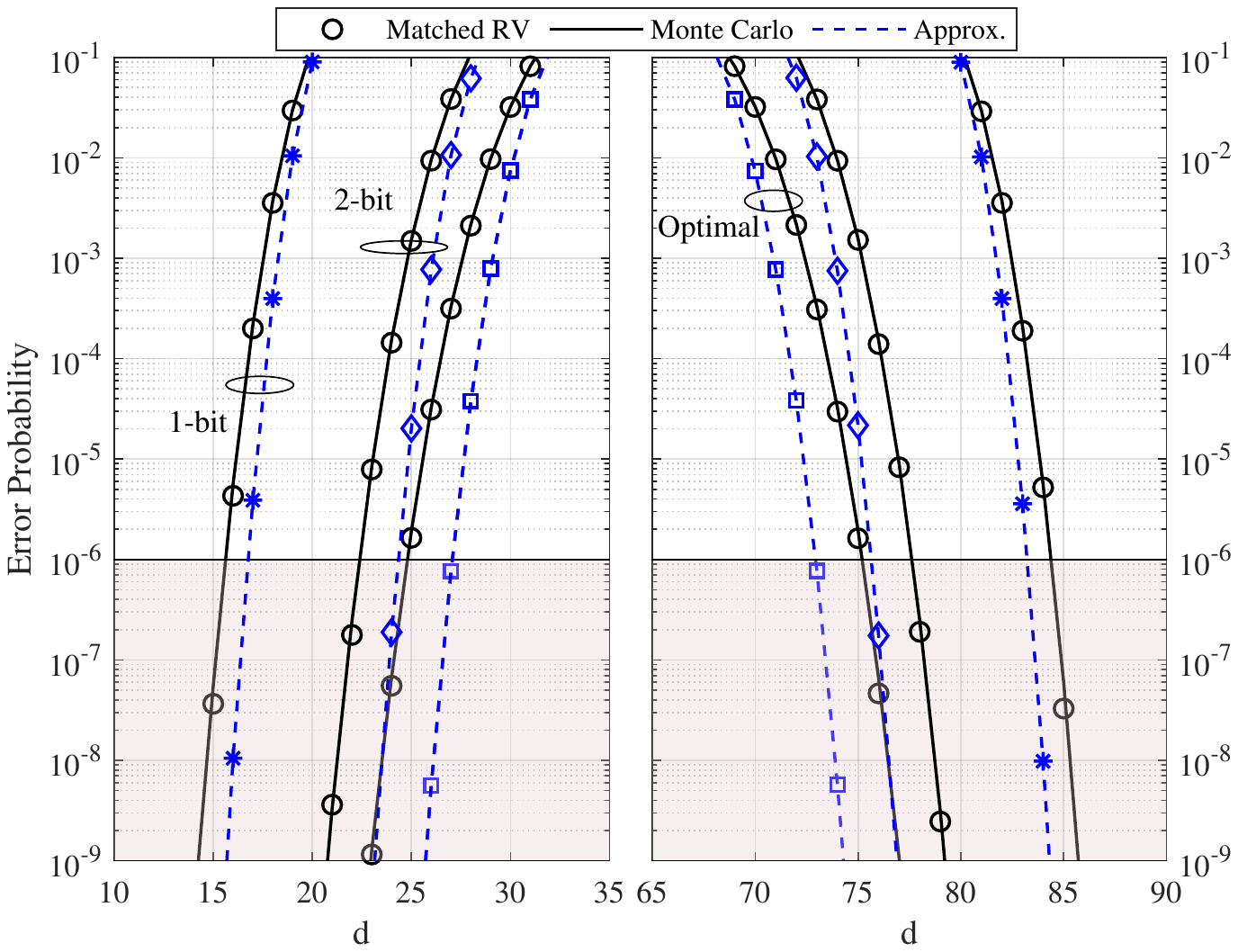}
		\caption{Average error probability performance in terms of changing the RIS location at ($d$, 10) for $N = 512$.}
		\label{fig:ErrorProbability_vs_distance}
    \end{figure}
    
    \textcolor{black}{In Fig. \ref{fig:Ergodic_rate_vs_distance_with_LoS_N4096} the average achievable rate is illustrated for $d \in [5,95]$, where RIS is located at ($d$, 10) and there is direct channel between AP and AC for $N=1024$. The curves are shown for different path loss exponents (PLE) in direct path to assess the impact of PLE on the accuracy of the proposed lowerbound for average FBL rate.} In contrast to the active relaying schemes where locating in the middle of transmitter and the receiver usually is an optimal choice to maximize the performance, we observe that the average achievable rate in the FBL regime will be maximized when the RIS is either close to transmitter or receiver. Additionally, there is a gap between the lower bound and exact value for all curves. Nevertheless, the gap does not change as the number of quantizer bits or the RIS location at ($d$, 10) change which confirms the suitability of the presented lower bound for comparison purposes and resource allocation algorithms.
   	\begin{figure*}[t]
   	\begin{subfigure}[b]{0.49\textwidth}  
		\centering
		\includegraphics[trim = 9.5cm 9.5cm 9.5cm 9.5cm,scale=0.6]{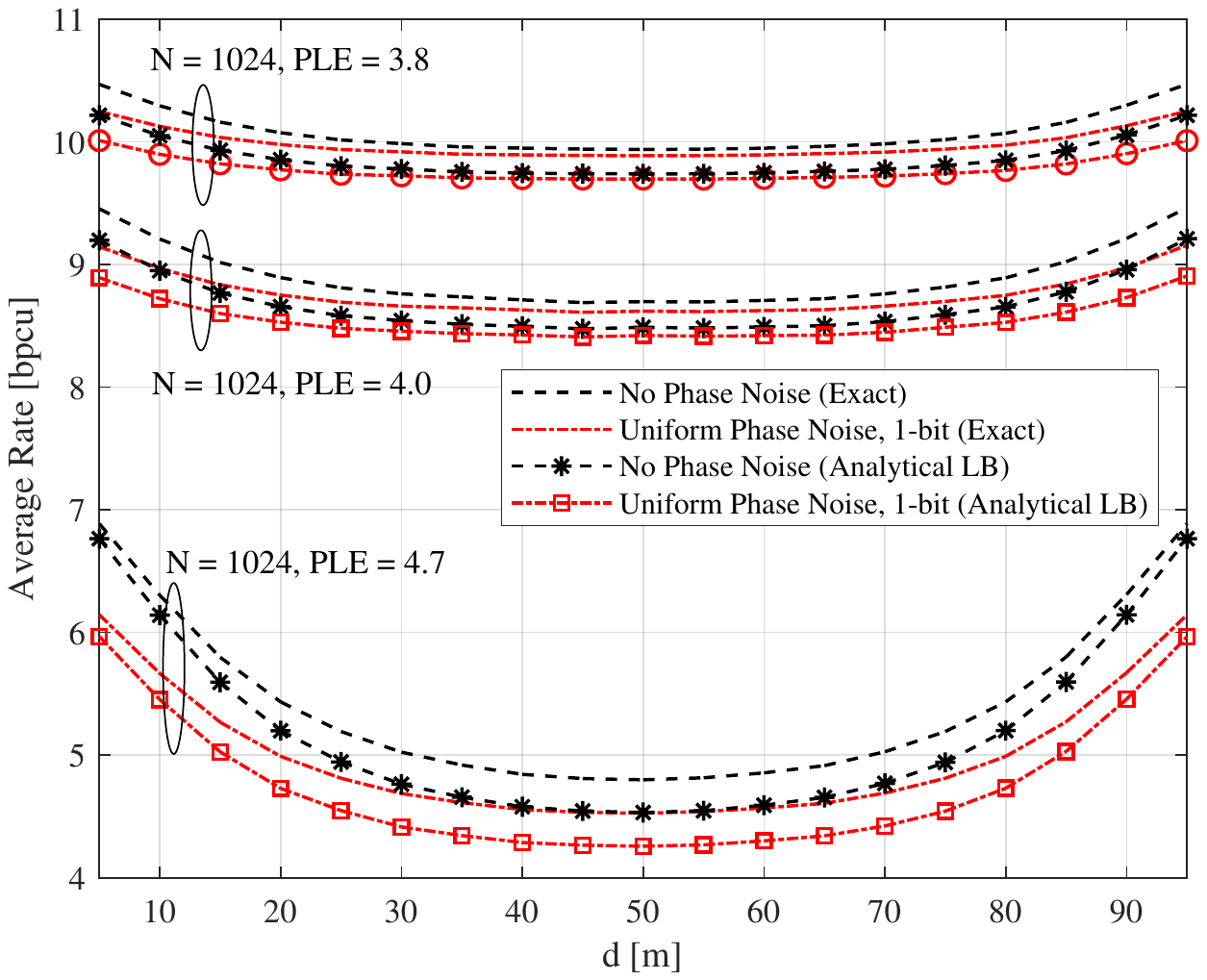}
		\caption{With direct channel.}
		\label{fig:Ergodic_rate_vs_distance_with_LoS_N4096}
    \end{subfigure}	
    \begin{subfigure}[b]{0.49\textwidth}  
		\centering
		\includegraphics[trim = 8.9cm 8.9cm 8.9cm 8.9cm,scale=0.6]{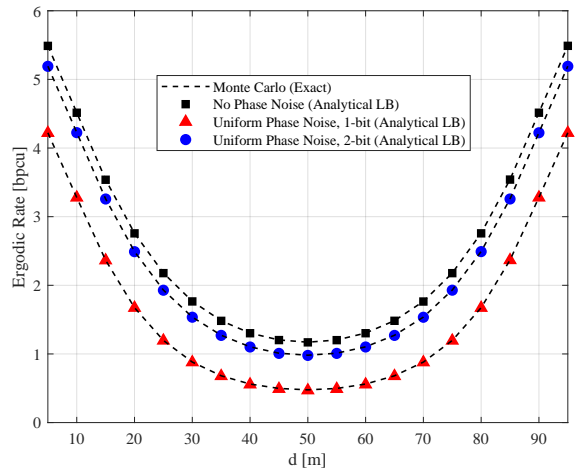}
		\caption{\textcolor{black}{Without direct channel.}}
		\label{fig:Ergodic_rate_vs_distance_without_LoS_N4096}
    \end{subfigure}
    \caption{The impact of changing the RIS location at ($d$, 10) on the average rate.}
    \end{figure*}
    \textcolor{black}{The average rate performance without direct channel is illustrated in Fig. \ref{fig:Ergodic_rate_vs_distance_without_LoS_N4096} where as shown in the curves there is a perfect match between the lower bound rate and the Monte Carlo simulations. Furthermore, the impact of phase error on the average rate improvement is the same.  To have a similar analysis of rate variation we see that when $d=50$ is changed to $d=95$ or $d=5$ the average rate is increased from $0.5$ ($1.0$) bpcu to $4.3$ ($5.3$) bpcu for 2-bit (1-bit) quantizer.} 
    
    In Fig. \ref{fig:AverageCH_use} the required number of channel uses as a function of RIS elements is illustrated in terms of quantizer  bits at the RIS \textcolor{black}{when target error probability is set to $10^{-9}$}. It is observed that, when the number of bits increases, the average channel blocklength will be reduced and asymptotically converges to the lower curve which is the case without noise at the RIS. There is a significant reduction when the number of quantization bits increases from one to two bits. This shows that in system level design considerations choosing a 2-bit quantizer will be beneficial and satisfactory than the complicated higher bit quantizers. Furthermore, given the channel blocklength $r=TW$ when the number of RIS elements increases the required number of channel blocklength to transmit the symbols will be reduced. \textcolor{black}{This means that as $N$ increases, the transmission duration $T$ is reduced given a fixed bandwidth $W$. Thus, RIS can be leveraged to achieve low latency transmissions for URLLC applications, demonstrating the applicability of RIS in FBL regime communications.}
    
  	\begin{figure}[t]
		\centering
		\includegraphics[trim = 8.5cm 8.5cm 8.5cm 8.5cm,scale=0.6]{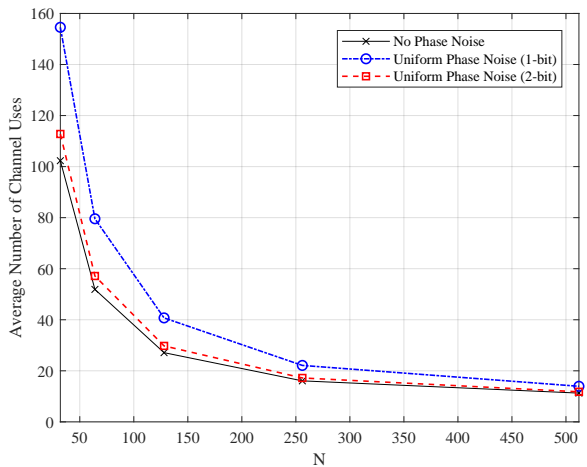}
		\caption{The average number of channel uses versus total number of RIS elements $N$.}
		\label{fig:AverageCH_use}
    \end{figure}
    
    Finally, the asymptotic behavior of the average square root of the channel dispersion, as well as the binomial approximation accuracy, is investigated in Fig. \ref{fig:ChannelDispersion}. As it is shown, the channel dispersion and its binomial approximation are well-matched when no phase error exists. The situation is different for the case when phase error exists and in a lower number of RIS elements, the accuracy is lower. Nevertheless, it can also be inferred that the channel dispersion asymptotically approaches to its upperbound when we have a sufficient number of RIS elements ($N$ approximately greater than 150) in both cases. This is because the received SNR increases so that the square root of channel dispersion will approach $\log_2(e)$. 
   	\begin{figure}[t]
		\centering
		\includegraphics[trim = 8.5cm 8.5cm 8.5cm 8.5cm,scale=0.6]{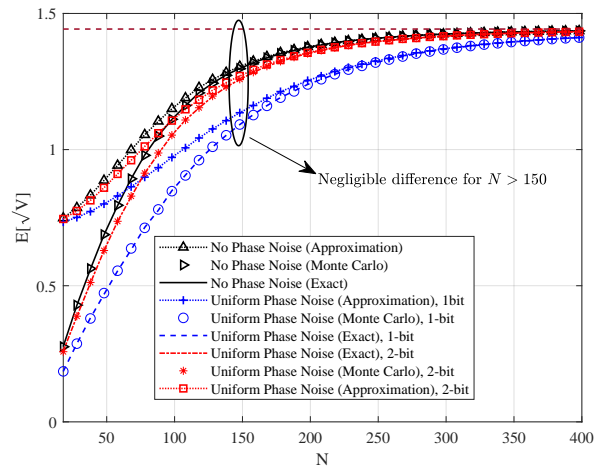}
		\caption{Average square root of channel dispersion ($\mathbb{E}[\sqrt{V(\gamma)}]$) and its asymptotic behaviour for large number of $N$.}
		\label{fig:ChannelDispersion}
    \end{figure}

    \section{Conclusion}
    \textcolor{black}{This paper analyzes the applicability of RIS for ensuring URLLC with FBL transmissions in a factory automation scenario.} We have presented the analytical derivation of the average achievable rate and we have analyzed the average error probability based on matching the received SNR to a Gamma RV. First, the analytical derivations of matching the SNR to a Gamma RV is presented whose parameters depend on the statistical mean and variance of the instantaneous SNR. Then, the average achievable rate is investigated \textcolor{black}{for the FBL regime} based on computing the related expectation with respect to proposed SNR distribution. The same analysis is performed to evaluate the average block error probability and channel blocklength. Next, the impact of phase error \textcolor{black}{resulting from} either quantization noise or hardware impairments is investigated in modeled SNR distribution parameters. The numerical results have shown that the RIS can be effectively employed in factory automation environments to \textcolor{black}{ensure high reliability and} reduce the error probability as a measure of reliability as well as the transmission latency as required by many URLLC applications. Furthermore, RIS is also found to significantly improve the average achievable rate. 
    
    As a future work, the analysis results in this paper can be leveraged for resource allocation problems in RIS-assisted URLLC networks to ensure systems reliability and maximizing average rate relying only on statistical measures of the channel. \textcolor{black}{On the other hand, as the size of the RIS grows, the identification of several energy/power consumption/dissipation sources, e.g., the utilized power at digital signal processing capabilities and the surface configuration power will be of paramount importance which  are interesting future research topics.} \textcolor{black}{Furthermore, the delay violation probability, distribution of the delay in URLLC systems and specifically the distribution of the channel uses as a function of the number of RIS elements are interesting future research directions where  the RIS technology plays a critical role  toward realizing URLLC systems' KPIs which highlights its applicability in factory automation environments.}
    \appendices
    
    \section{Computing \texorpdfstring{$\mathbb{E}[X]$}{} and \texorpdfstring{$\mathbb{E}[X^2]$}{} when \texorpdfstring{$\phi_n\neq0$}{} for \texorpdfstring{$\forall n \in \mathcal{N}$}{}}  
    
    \label{appendixA}
    Let us rewrite the random variable $X$ as 
    \begin{flalign}
        X = & \left|\left|h^{\text{AP}}_{\text{AC}}\right| + \sum_{n=1}^{N} \left|[\textbf{h}_{\text{RIS}}^{\text{AP}}]_n\right| \left| [\textbf{h}_{\text{AC}}^{\text{RIS}}]_n\right|e^{\mathrm{j}\phi_n}\right|^2, 
    \end{flalign}
    where $[.]_n$ denotes the n$^\text{th}$ element of a vector and $\phi_n =  \angle{h^{\text{AP}}_{\text{AC}}}- \angle [\textbf{h}_{\text{AC}}^{\text{RIS}}]_n + \angle [\textbf{h}_{\text{RIS}}^{\text{AP}}]_n + \theta_n$. In order to obtain the parameters of matched Gamma distribution, we should determine the shape and rate parameters. To do so, we first calculate the expected mean value of RV $X$ noting that for a RV $h \sim \text{Rayleigh}(\sigma)$ it holds  $\mathbb{E}[h] = \sqrt{\frac{\pi}{2}}\sigma$ which is given by
    \begin{flalign}
        \mathbb{E}&[X] = \eta^{\text{AP}\rightarrow \text{AC}} + \sum_{n=1}^{N} \eta^{\text{AP}\rightarrow\text{RIS}}_n \eta^{\text{RIS}\rightarrow\text{AC}}_n  \nonumber \\ & + \frac{\pi^2}{16}\sum_{n=1}^{N} \sum_{\substack{m=1\\ m \neq n}}^{N} \eta^{\text{RIS} \rightarrow\text{AP}}_n \eta^{\text{AC}\rightarrow\text{RIS}}_m \cos(\phi_n-\phi_m) + \nonumber \\ & \frac{\pi}{4}\sqrt{\pi\eta^{\text{AP}\rightarrow \text{AC}}} \sum_{n=1}^{N} \sqrt{\eta^{\text{RIS} \rightarrow\text{AP}}_n \eta^{\text{AC}\rightarrow\text{RIS}}_n} \cos(\phi_n),
        \label{expected_value_2}
    \end{flalign}
    where we assumed the channel responses are mutually independent. By neglecting the impact of RIS surface dimensions on large scale fading we will have
    \begin{flalign}
        \mathbb{E}&[X] = \eta^{\text{AP}\rightarrow \text{AC}} + N \eta^{\text{AP}\rightarrow\text{RIS}} \eta^{\text{RIS}\rightarrow\text{AC}} \nonumber \\ & + \frac{\pi^2}{16}\eta^{\text{RIS} \rightarrow\text{AP}} \eta^{\text{AC}\rightarrow\text{RIS}} \sum_{n=1}^{N} \sum_{\substack{m=1\\ m \neq n}}^{N}  \cos(\phi_n-\phi_m)  \nonumber \\ & +\frac{\pi}{4}\sqrt{\pi\eta^{\text{AP}\rightarrow \text{AC}}\eta^{\text{RIS} \rightarrow\text{AP}} \eta^{\text{AC}\rightarrow\text{RIS}}} \sum_{n=1}^{N} \cos(\phi_n),
        \label{expected_value_3}
    \end{flalign}
    where $\eta^{\text{AP}\rightarrow \text{AC}} = \eta^{\text{AP}\rightarrow \text{AC}}_n$, $\eta^{\text{AP}\rightarrow\text{RIS}}=\eta^{\text{AP}\rightarrow\text{RIS}}_n$ and $\eta^{\text{RIS}\rightarrow\text{AC}}=\eta^{\text{RIS}\rightarrow\text{AC}}_n$  $\forall n \in \mathcal{N}$.

    In what follows we continue by computing the expected value of $X^2$ which is defined as
    \begin{flalign}
        \mathbb{E}[X^2] = \mathbb{E}\left[\left||h^{\text{AP}}_{\text{AC}}| + \sum_{n=1}^{N} \big|[\textbf{h}_{\text{RIS}}^{\text{AP}}]_n\big| \big| [\textbf{h}_{\text{AC}}^{\text{RIS}}]_n\big|e^{\mathrm{j}\phi_n}\right|^4\right],
    \end{flalign}
    for notation simplicity we define 
    \begin{subequations}
        \begin{flalign}
        c_0  &\coloneqq |h^{\text{AP}}_{\text{AC}}|, \\
        a_n &\coloneqq \big|[\textbf{h}_{\text{RIS}}^{\text{AP}}]_n\big|, \\
        b_n &\coloneqq \big| [\textbf{h}_{\text{AC}}^{\text{RIS}}]_n\big|.
        \end{flalign}
    \end{subequations}
    The binomial expansion of the expression inside expectation yields
    \begin{flalign}
        \mathbb{E}&\Big[\Big|c_0 + \sum_{n=1}^{N} a_n b_n e^{\mathrm{j}\phi_n}\Big|^4\Big]
        = \mathbb{E}\Bigg[c_0^4 + 2c_0^2\sum_{n=1}^{N}a_n^2b_n^2(1+2\cos^2(\phi_n))  \nonumber \\ & + 2 c_0^2\sum_{n=1}^{N}\sum_{\substack{m=1\\ m\neq n}}^{N}a_nb_na_mb_m(\cos(\phi_n-\phi_m)+2\cos(\phi_n)\cos(\phi_m)) \nonumber \\ 
        & + 4c_0^3\sum_{n=1}^{N}a_nb_n\cos(\phi_n) + \sum_{n=1}^{N}\sum_{\substack{m=1\\ m\neq n}}^{N}a_n^2b_n^2a_m^2b_m^2 + \sum_{n=1}^{N}a_n^4b_n^4  \nonumber \\ & + 2 \sum_{n=1}^{N}a_n^2b_n^2\sum_{n=1}^{N}\sum_{\substack{m=1\\ m\neq n}}^{N}a_nb_na_mb_m\cos(\phi_n-\phi_m) \nonumber \\
        &  + 4c_0\sum_{n=1}^{N}\sum_{\substack{m=1\\ m\neq n}}^{N}a_n^2b_n^2a_mb_m\cos(\phi_m) + 4c_0\sum_{n=1}^{N}a_n^3b_n^3\cos(\phi_n)  \nonumber \\ & + \Big(\sum_{n=1}^{N}\sum_{\substack{m=1\\ m\neq n}}^{N}a_nb_na_mb_m\cos(\phi_n-\phi_m)\Big)^2 \nonumber \\
        & + 4 c_0 \sum_{n=1}^{N}a_nb_n\cos(\phi_n)\sum_{n=1}^{N}\sum_{\substack{m=1\\ m\neq n}}^{N}a_nb_na_mb_m\cos(\phi_n-\phi_m)\Bigg],
        \label{appndx_1}
    \end{flalign}
    since $\mathbb{E}[c_0] = \frac{\sqrt{\pi\eta^{\text{AP} \rightarrow \text{AC}}}}{2}$, $\mathbb{E}[a_n] = \frac{\sqrt{\pi\eta^{\text{AP} \rightarrow \text{RIS}}}}{2}$ and $\mathbb{E}[b_n] = \frac{\sqrt{\pi\eta^{\text{RIS} \rightarrow \text{AC}}}}{2}$ and also noting that for a RV $h\sim$Rayleigh$(\sigma)$ we have 
    \begin{subequations}
        \begin{flalign}
            \mathbb{E}[h] = & \sqrt{\frac{\pi}{2}}\sigma,  \\ 
            \mathbb{E}[h^2] = & 2\sigma^2,  \\ 
            \mathbb{E}[h^3] = & 3\sqrt{\frac{\pi}{2}}\sigma^3, \\ 
            \mathbb{E}[h^4] = & 8\sigma^4.
        \end{flalign}
    \end{subequations}
    Based on above, and after some mathematical manipulations the equation \eqref{appndx_1} will be written as given in Eq. \eqref{appndx_2} on top of the next page.
    \begin{figure*}
    \begin{flalign}
        &\mathbb{E}\Big[\Big|c_0 + \sum_{n=1}^{N} a_n b_n e^{\mathrm{j}\phi_n}\Big|^4\Big] = 2\varsigma^2 + \nonumber \\
        &  \varsigma\varrho\vartheta\bigg(2N+4\sum_{n=1}^{N}\cos^2(\phi_n)  + \frac{\pi^2}{8}\sum_{n=1}^{N}\sum_{\substack{m=1\\ m\neq n}}^{N}\left(\cos(\phi_n-\phi_m)+2\cos(\phi_n)\cos(\phi_m)\right) \bigg) +    \frac{3\pi}{4}\sqrt{\pi\varsigma^3\varrho\vartheta} \sum_{n=1}^{N}\cos(\phi_n)  \nonumber \\ 
        & 
        +  \varrho^2\vartheta^2 \bigg(N(N+3)+\frac{\pi^2(2N+5)}{16}\sum_{n=1}^{N}\sum_{\substack{m=1\\ m\neq n}}^{N}\cos(\phi_n-\phi_m) +2\sum_{n=1}^{N}\sum_{\substack{m=1\\ m\neq n}}^{N}\cos^2(\phi_n-\phi_m) \nonumber \\ 
        & \quad\quad  + \frac{\pi^2}{8} \sum_{n=1}^{N}\sum_{\substack{m=1\\ m\neq n}}^{N}\sum_{\substack{n'=1\\ n'\neq n,m}}^{N}\cos(\phi_n-\phi_m)\cos(\phi_{n'}-\phi_{n}) +  \frac{\pi^2}{8} \sum_{n=1}^{N}\sum_{\substack{m=1\\ m\neq n}}^{N}\sum_{\substack{n'=1\\ n'\neq n,m}}^{N}\cos(\phi_n-\phi_m)\cos(\phi_{n'}-\phi_{m}) \nonumber \\ 
        & \quad\quad  +\frac{\pi^4}{2^8} \sum_{n=1}^{N}\sum_{\substack{m=1\\ m\neq n}}^{N}\sum_{\substack{n'=1,\\n'\neq n,m}}^{N}\sum_{\substack{m'=1,\\m'\neq n, m,n'}}^{N}\cos(\phi_n-\phi_m)\cos(\phi_{n'}-\phi_{m'}) \bigg)
        \nonumber \\ 
        &
        + \sqrt{\varsigma\varrho^3\vartheta^3}\pi \sqrt{\pi}\bigg( \frac{4N+5}{8} \sum_{n=1}^{N}\cos(\phi_n)+ \frac{1}{2}\sum_{n=1}^{N}\sum_{\substack{m=1\\ m\neq n}}^{N}\cos(\phi_n-\phi_m)\cos(\phi_n) \nonumber \\ 
        &
        \quad\quad +\frac{1}{2}\sum_{n=1}^{N}\sum_{\substack{m=1\\ m\neq n}}^{N}\cos(\phi_n-\phi_m)\cos(\phi_m) +\frac{\pi^2}{32}\sum_{n=1}^{N}\sum_{\substack{m=1\\ m\neq n}}^{N}\sum_{\substack{n'=1\\ n'\neq n, m}}^{N}\cos(\phi_n-\phi_m)\cos(\phi_{n'}) \bigg),
        \label{appndx_2}
    \end{flalign}
    \hrule
    \end{figure*}
    where $\mathbb{E}[c_0^2]=\eta^{\text{AP} \rightarrow \text{AC}}=\varsigma $, $ \mathbb{E}[a_n^2]=\eta^{\text{AP} \rightarrow \text{RIS}}=\varrho$ and $ \mathbb{E}[b_n^2]= \eta^{\text{RIS} \rightarrow \text{AC}}=\vartheta$. 

    \section{Computing \texorpdfstring{$\mathbb{E}[X]$}{} and \texorpdfstring{$\mathbb{E}[X^2]$}{} when \texorpdfstring{$\phi_n=0$}{} for \texorpdfstring{$\forall n \in \mathcal{N}$}{}}
    \label{appendixB}
    In a special case where the phase adjustment is perfectly done at the RIS we have $\phi_n = 0$, $\forall n \in \mathcal{N}$ which yields
    \begin{flalign}
        &\mathbb{E}[X] = \eta^{\text{AP}\rightarrow \text{AC}} + N \eta^{\text{AP}\rightarrow\text{RIS}} \eta^{\text{RIS}\rightarrow\text{AC}} + \nonumber \\ 
        & \frac{\pi^2N(N-1)}{16}\eta^{\text{AP} \rightarrow\text{RIS}} \eta^{\text{RIS}\rightarrow\text{AC}} + \frac{\pi N}{4}\sqrt{\pi\eta^{\text{AP}\rightarrow \text{AC}}\eta^{\text{AP} \rightarrow\text{RIS}} \eta^{\text{RIS}\rightarrow\text{AC}}},
        \label{appndx_b_eq1}
    \end{flalign}
    and
    \begin{flalign}
        \mathbb{E}[X^2]= & \mathbb{E}\left[\left|c_0 + \sum_{n=1}^{N} a_n b_n\right|^4\right] = 2\varsigma^2 + \varsigma\varrho\vartheta N\big(6  + \frac{3(N-1)\pi^2}{8} \big) \nonumber \\ 
        &  +  \frac{3N\pi^{1.5}}{4}\sqrt{\varsigma^3\varrho\vartheta} +
        \frac{\varrho^2\vartheta^2N}{256} \Big(\pi ^4 (N-3) (N-2) (N-1) \nonumber \\ 
        & + 48 \pi ^2 (2 N-1) (N-1)+768 N + 256 \Big) 
        \nonumber \\ 
        &
        + \sqrt{\varsigma\varrho^3\vartheta^3}\frac{N \pi^{1.5} }{32} \left( \pi ^2 (N-2) (N-1)+48 N-12 \right).
        \label{appndx_3}
    \end{flalign}
    where $\varsigma = \eta^{\text{AP} \rightarrow \text{AC}}$, $\varrho = \eta^{\text{AP} \rightarrow \text{RIS}}$ and $\vartheta = \eta^{\text{RIS} \rightarrow \text{AC}}$.

    \textcolor{black}{\section{Proof of Remark \ref{remark_channel_use}}
    \label{appendixC}
    If we define $A\coloneqq\frac{\mathcal{C}_1}{\mathcal{C}_2Q^{-1}(\varepsilon)}=\frac{\mathbb{E}[\log_2(1+\gamma)]}{Q^{-1}(\varepsilon)\mathbb{E}[\sqrt{\text{V}(\gamma)}]}$ the average channel use is reformulated as follows
    \begin{flalign}
        \bar{r} = \frac{L}{\mathcal{C}_1} + \frac{1}{4A^2}+\frac{1}{2A}\sqrt{\frac{1}{A^2}+\frac{4L}{\mathcal{C}_1}}, \label{r_bar_appx}
    \end{flalign}
    then, noting that $A$ is an increasing function in terms of number of RIS elements (this can be readily proved by firstly noting that the SNR $\gamma$ increases with respect to RIS elements $N$ secondly by computing the first derivative of $f(\gamma)=\frac{\log_2(1+\gamma)}{\sqrt{\text{V}(\gamma)}}$ with respect to $\gamma$ which is $f'(\gamma)\geq0$), all the terms in the expression given in \eqref{r_bar_appx} will correspond to  be decreasing. This completes the proof.}
    
	\bibliographystyle{IEEEtran}
    \bibliography{IEEEabrv,refs}

	
\end{document}